\documentclass[12pt]{amsart}
\usepackage{amsmath,amsthm,amssymb,mathtools,yhmath}
\usepackage{graphicx,adjustbox}
\usepackage{stmaryrd}
\usepackage{amsfonts}
\usepackage{mathrsfs}
\usepackage{enumerate, url}
\usepackage[margin=1in]{geometry}
\usepackage{caption}
\usepackage{algorithm}
\usepackage{algpseudocode}
\usepackage{mleftright}
\usepackage{hyperref}
\usepackage{tikz}
\usepackage{tabu}
\usepackage{float}
\usepackage{indentfirst}
\usepackage{mathdots} 
\usepackage[all]{xy}
\usepackage{subfigure}
\usepackage{tikz-cd}  
\usetikzlibrary{positioning}
\usepackage{xcolor}
\hypersetup{
	colorlinks,
	linkcolor={red!50!black},
	citecolor={blue!50!black},
	urlcolor={blue!80!black}
}
\usepackage{hyperref}
\usepackage{tabu}
\usepackage{subfigure}
\newtheorem{theorem}{Theorem}[section]
\newtheorem{proposition}[theorem]{Proposition}
\newtheorem{proposition/definition}[theorem]{Proposition/Definition}
\newtheorem{lemma}[theorem]{Lemma}
\newtheorem{corollary}[theorem]{Corollary}

\theoremstyle{definition}
\newtheorem{definition}[theorem]{Definition}

\newtheorem{example}[theorem]{Example}

\theoremstyle{remark}

\newtheorem{remark}[theorem]{Remark}

\newcommand{\Rmnum}[1]{\expandafter\@slowromancap\Romannumeral #1@}

\DeclareMathOperator{\im}{im}

\DeclareMathOperator{\diag}{diag}
\DeclareMathOperator{\spn}{span}

\DeclareMathOperator{\Ass}{Ass}
\DeclareMathOperator{\Lcm}{Lcm}

\DeclareMathOperator{\Frac}{Frac}
\DeclareMathOperator{\Sqrt}{Sqrt}
\begin{document}
	\title{Commutative algebra-enhanced topological data analysis}
	\author[C.S.~Hu]{Chuan-Shen Hu}
	\address{School of Physical and Mathematical Sciences,
		Nanyang Technological University,
		21 Nanyang Link, Singapore 637371}
	\email{chuanshenhu.official@gmail.com}
	\author[Y.~Wang]{Yu Wang}
	\address{KLMM, Academy of Mathematics and Systems Science, Chinese Academy of Sciences, Beijing 100190, China}
	\email{wangyu2020@amss.ac.cn}
	\author[K.L.~Xia]{Kelin Xia}
	\address{School of Physical and Mathematical Sciences,
		Nanyang Technological University,
		21 Nanyang Link, Singapore 637371}
	\email{xiakelin@ntu.edu.sg}
	\author[K.~Ye]{Ke Ye }
	\address{KLMM, Academy of Mathematics and Systems Science, Chinese Academy of Sciences, Beijing 100190, China}
	\email{keyk@amss.ac.cn}
	\author[Y.P.~Zhang]{Yipeng Zhang}
	\address{School of Physical and Mathematical Sciences,
		Nanyang Technological University,
		21 Nanyang Link, Singapore 637371}
	\email{yipeng001@e.ntu.edu.sg}
	\date{\today}
	
	\begin{abstract}
		Topological Data Analysis (TDA) combines computational topology and data science to extract and analyze intrinsic topological and geometric structures in data set in a metric space. While the persistent homology (PH), a widely used tool in TDA, which tracks the lifespan information of topological features through a filtration process,  has shown its effectiveness in applications, 
		it is inherently limited in homotopy invariants 
		and overlooks finer geometric and combinatorial details. To bridge this gap,  we introduce two novel commutative algebra-based frameworks which extend beyond homology by incorporating tools from computational commutative algebra : (1) \emph{the persistent ideals} derived from the decomposition of algebraic objects associated to simplicial complexes, like those in theory of edge ideals and Stanley--Reisner ideals, which will provide new commutative algebra-based barcodes and offer a richer characterization of topological and geometric structures in filtrations. 
		(2)\emph{persistent chain complex of free modules} associated with traditional persistent simplicial complex by labelling each chain in the chain complex of the persistent simplicial complex with elements in a commutative ring, which will enable us to detect local information of the topology via some pure algebraic operations. \emph{Crucially, both of the two newly-established framework can recover topological information got from conventional PH and will give us more information.} Therefore, they provide new insights in computational topology, computational algebra and data science.
	\end{abstract}
	\maketitle
	\section{Introduction}

		Topological Data Analysis (TDA) is an interdisciplinary field that integrates computational topology, homological algebra, and data science to uncover and analyze intrinsic topology and geometry within high-dimensional, noisy, and complex datasets~\cite{carlsson2009topology, ghrist2008barcodes, Barcodes_for_shapes, singh2007topological, Epstein_2011}. By leveraging algebraic topology, TDA provides a robust framework for detecting and quantifying topological structures. Due to its ability to extract meaningful topological information from structured and unstructured data, TDA has found applications across various scientific and engineering areas such as bioinformatics~\cite{HOZUMI2024115842,Liu2022,Bhaskar2023}, material science~\cite{nakamura2015persistent,hiraoka2016hierarchical,murakami2019ultrahigh}, and deep and machine learning models~\cite{zhu2023tidal,Liu2022,chen2021sars}.
		
		A key tool in TDA, known as \textit{persistent homology} (PH), tracks the evolution of these features throughout a filtration, offering a multi-scale representation of the point cloud data. To summarize and extract features from persistent homology, the persistence barcode (PB) is a widely-used representation that encodes the lifespan of topological features across a filtration~\cite{barannikov1994framed,frosini1992measuring,ghrist2008barcodes,Barcodes_for_shapes,mischaikow2013morse,robins1999towards}. It serves as a complete homological invariant, detecting and recording homological information within topological filtrations~\cite{barannikov1994framed}. In particular, homological changes—such as the birth and death of loop structures—are effectively captured as intervals called \textit{bars} in a barcode. Generally, longer bars correspond to significant topological structures, while shorter bars are often considered less informative or associated with noise. With their ability to measure shapes and structures using size functions~\cite{frosini1992measuring}, capture geometric descriptors in Morse theory~\cite{mischaikow2013morse}, and infer the topology and geometry of attractors in dynamical systems~\cite{robins1999towards}, persistence barcodes provide a compact yet insightful summary of the topological evolution of data. This makes them particularly valuable in applications where shape and connectivity play a crucial role.
		
		However, despite their effectiveness, topology- and homology-based approaches have certain limitations. While these methods---such as PBs---capture essential homological features, including the birth and death of connected components, loops, and voids, they often fail to preserve finer geometric and combinatorial information present in the data. For instance, distinct datasets with different underlying geometric structures can yield similar or even identical PBs, limiting their discriminative power in applications that require more nuanced representations. Moreover, tools like PBs are typically represented as unordered multisets of intervals, which present challenges for statistical comparison, vectorization, and integration into machine learning workflows. In particular, local information---such as the specific locations and involved data points associated with persistence intervals---is often omitted. These limitations have motivated the development of more expressive tools that incorporate richer geometric detail, including algebraically enriched frameworks such as combinatorial Hodge Laplacian-based methods \cite{HORAK2013303,guelen2023generalization,PersistentLaplacians,qiu2023persistent}, sheaf-theoretic approaches \cite{curry2015topological,curry2016discrete,Curry2018,hansen2019toward}, and commutative algebraic constructions \cite{ren2018weighted, liu2023persistent, faulstich2024algebraic,Schenck2022algebraic}. These extensions aim to capture deeper structural information and improve compatibility with modern data analysis pipelines.
		To address these challenges, this work introduces a novel framework grounded in computational commutative algebra for analyzing topological filtrations and their persistence.

		\subsection*{Main Results } This paper introduces two different commutative algebra-based methods for TDA. The first method is about the \emph{persistent ideals.} We first introduce some commutative algebraic objects--edge ideals and Stanley-Reisner ideals--linked to simplicial complex. Then we prove a persistent version of the decomposition of the ideals (cf. Theorem~\ref{Proposition: lemma for the birth-death information}  and therefore we can define \emph{the persistent associated primes} as new commutative algebra-based persistence barcodes. Then we prove that via some algebraic operation, the newly-defined barcodes based on the Stanley-Reisner ideals will re-establish the information got from traditional PB (cf. Theorem~\ref{Persistent affine variety covers persistent Betti number} and Theorem~\ref{Persistent affine variety covers persistent Betti number: dicrete version}), and thus shows that they encode finer combinatorial and geometric features. It needs to be mentioned that we really provide a much comprehensive method for constructing the persistent ideals via Theorem~\ref{Proposition: lemma for the birth-death information} and the method is also valid for other ideals  defined by the simplicial complex such as those mentioned in \cite{Hibi2011monomial, Hibi2018binomial}--though the information we get may vary.
		
		The second method consists of the \emph{persistent chain complex of free modules} associated with the simplicial one by labeling each chain with certain elements in a \emph{unique factorization domain }(UFD). These complex can be viewed as a "twisted version" of the tranditional one, like \emph{the weighted persistent homology}~\cite{dawson1990homology, ren2018weighted} and also mentioned in \cite[Chapter 8]{Schenck2022algebraic} and \cite{Carlsson2009} in the form \emph{multiparameter persistence}. We prove some equivalence between this construction and the original persistent chain complex and illustrate some new features that can be read from the new barcodes. To be specific, we can prove that it can also recover the information obtained from the traditional PB by two different ways (cf. Theorem~\ref{equivalence of the two chain complexes via evaluation} and Theorem~\ref{equivalence of the two chain complexes via localization}). Furthermore, by adjusting the label elements in the UFD, we may get the information from \emph{any full subcomplex}(cf. Corollary~\ref{corollary: local information via evaluation} and Corollary~\ref{corollary: local information via localization}), which illustrates the "local inforamtion". Thus by flexible choice of the label elements, we can get vast topological features that are not easy to obtain through the traditional method.  More importantly, since we can use different UFD for labelling, this method offers a unified framework from algebraic geometric viewpoints.


	\section{Preliminary}
	Readers may be assumed to be familiar with basic abstract algebra, especially basic theory about rings and modules, such as the tensor product.  Books like \cite{Lang,Jacobsonn1989vol1,Jacobsonn1989vol2} will be good references for it.  The basic knowledge of general topology and algebraic topology, at least those used in traditional topological data analysis, is obviously required. Books like \cite{hatcher2002AT,munkres1984elements,massey1991AT} are good references for topology. 
	
		
		\subsection{Notations and Terminologies}\label{Section: Notations and Terminologies}
		
		The notations \( \mathbb{N}, \mathbb{Z}, \mathbb{Q}, \mathbb{R} \), and \( \mathbb{C} \) represent the sets of natural numbers, integers, rational numbers, real numbers, and complex numbers, respectively. The number zero is included in the set of natural numbers 
		and the set of all positive integers is denoted by \( \mathbb{Z}_+ \).  For any \( n \in \mathbb{Z}_+ \), the permutation group on the set \( [n] := \{1, \dots, n\} \) is denoted by \( S_n \), and the sign of a permutation \( \phi \in S_n \) is written as \( \operatorname{sgn}(\phi) \). The power set of \( [n] \), denoted by \( 2^{[n]} \), consists of all subsets of \( [n] \).

		All rings $R$ considered in this paper are assumed to be commutative with a multiplicative identity $1_R$ (or simply $1$), unless stated otherwise. Typical examples include the ring of integers $\mathbb{Z}$, fields $\mathbb{F}$, polynomial rings $\mathbb{S}_n := \mathbb{F}[x_1, \dots, x_n]$ over a field $\mathbb{F}$ with $n$ variables, and the corresponding rational function fields $\mathbb{F}(x_1, \dots, x_n)$. 
		
		For any subset \( S \) of a ring $R$ , the ideal generated by \( S \) is denoted by \( \langle S \rangle_{R} \), or simply \( \langle S \rangle \) when the ring $R$ is clear from the context. Specifically, the ideal $\langle S \rangle$ consists of all finite sums of the form $\sum_{i = 1}^m r_is_i$ where $m \in \mathbb{N}$, $r_i \in R$, and $s_i \in S$. In particular, if $S = \{ s_1, s_2, \dots, s_k \}$ is a finite set with $k$ elements, the ideal $\langle S \rangle$ is also denoted by $\langle s_1, s_2, \dots, s_k \rangle$. 
		
		As for the ring $\mathbb{S}_n$, we denote $x^{\alpha}:=x_1^{\alpha_1}\cdots x_{n}^{\alpha_n}$ for $\alpha\in \mathbb{N}^n$. And for a nonempty set $\sigma\subset \{1,2,\dots,n\}$, we denote $x_{\sigma}:=\Pi_{i\in \sigma}x_i$. By convention, $x_{\emptyset}:=1$.
		
		The topological terminations we use here are mainly from simplicial homology in the combinatoric way. So an (\textit{abstract}) \textit{simplicial complex} over $[n]:=\{1, 2, \dots, n\}$ means a nonempty collection $\Delta$ of subsets $\sigma \in 2^{[n]}$ satisfying the following property: if $\sigma \in \Delta$ and $\emptyset\neq \tau \subset \sigma$, then $\tau \in \Delta$. By convention, we require that a simplicial complex does always NOT contain $\emptyset$.. The \textit{dimension} of a simplex $\sigma \in \Delta$, denoted as $\operatorname{dim}(\sigma)$, is defined as $\#\sigma - 1$.
	On the other hand, we will use the term "\emph{reduced simplicial complex}" for the complex containing $\emptyset$ as the \emph{$-1$ dimensional face}, and we omit the case when $\Delta=\{\emptyset\}$. This is coherent with the \emph{reduced homology}. 
	
	By \emph{simplex (or a facet) over $W\subset \{1,\dots,n\}$}, we mean the simplicial complex defined by all the subsets of $W$ except for the empty set.  And the \emph{reduced simplex} for the one containing $\emptyset$.  
	For \emph{subcomplex} $\Delta'$ of a simplicial complex $\Delta$, we mean that $\Delta'\subset \Delta$ as the family of subsets and $\Delta'$ itself is a simplicial complex. The $0-$dimensional faces are called the \emph{vertices} as usual. For a simplicial complex $\Delta$ over $\{1,2,\dots,n\}$ and $W\subset \{1,2,\dots,n\}$ contained in the set of vertices of $\Delta$,  we use the term \emph{the full subcomplex} $\Delta'$ of $\Delta$ over $W$ to denote the subcomplex $\Delta':=\{\sigma\in \Delta|\sigma\subset W \}$.


		\begin{example}[Clique complexes]
			Let \( G = (V, E) \) be an undirected simple graph. A subset \( \sigma \subset V \) is called a \textit{clique} of $G$ if $\# \sigma =1$ or \( \{ u, v \} \in E \) for every \( u, v \in \sigma \) . The \textit{clique complex} of \( G \), denoted by \( \Delta(G) \), is defined as a simplicial complex consisting of all cliques in \( G \), which can be proved indeed a simplicial complex over $V$.
		\end{example}

	\subsection{ A brief review of traditional TDA}
	Let's have a quick review of traditional topological data analysis. Given a data cloud $\mathfrak{D}:=\{v_1,\dots,v_n\}$ in a metric space, we want to find some topological information to describe it. The usual topological information we use is the \emph{persistent homology} (PH) from a filtration. For any $\epsilon>0$, denote the $\epsilon-$ball centered at $v$ as $B_{\epsilon}(v)$. We can construct a graph $G_\epsilon=\{V_\epsilon,E_\epsilon\}$ as follows:
	$V_\epsilon=\{1,\dots,n\}$ and $\{i,j\}\in E_\epsilon$ if and only if $d(v_i,v_j)<\epsilon$, i.e., the two $\frac{\epsilon}{2}-$balls centered at $v_i$ and $v_j$ intersect.  Then we can construct a simplicial complex linked with the graph. 
	
	The construction of the simplicial complex is not unique and here we use the \emph{clique complex} of the graph as an example, which is called the \emph{Vietoris-Rips complex}. By definition, the Vietoris-Rips complex is the simplicial complex $\Delta_\epsilon$ consisting of all the clique with $k+1-$vertices of the graph $G_{\epsilon}$ as its $k-$dimensional faces. With a simplicial complex as a topological space, we can construct a chain complex as follows:
	\begin{equation}\label{equation of the chain complex of persistent homology}
		\begin{tikzcd}
			\cdots\arrow[r,"\delta_{k+1}^\epsilon"]&C_k^\epsilon\arrow[r,"\delta_{k}^\epsilon"]&\cdots \arrow[r,"\delta_{2}^\epsilon"]&C_1^\epsilon\arrow[r,"\delta_1^\epsilon"]& C_0^{\epsilon}\arrow[r,"\delta_0^\epsilon"]& 0,
		\end{tikzcd}
	\end{equation}
	where the \emph{chain group} $C_{k}^\epsilon$ is the free abelian group/$\mathbb{F}-$vector space generated by the $k-$ dimensional faces of $\Delta_\epsilon$ and the linear map $\delta_k^\epsilon:C_k^\epsilon\longrightarrow C_{k-1}^\epsilon$ is defined as $\delta_k^{\epsilon}(\sigma):=\sum_{u=1}^{k+1} (-1)^u \sigma\setminus\{i_u\}$ where $\sigma=\{i_1,\dots,i_{k+1}\}\in\Delta_{\epsilon}$ with $i_1<\dots<i_{k+1}$, and extended to the whole chain group linearly. Then we can calculate the $k-$th \emph{(persistent) homology group} $H_k^\epsilon:=\ker(\delta_k^{\epsilon})/\im(\delta_{k+1}^\epsilon)$. Here the word "persistent" means we can get a sequence of homology group when $\epsilon$ increases continuously or discretely. Usually, we use the \emph{rank}/\emph{dimension} of the homology groups as numerical indicators for the topological information. The rank/dimension of the $k-$th homology group is a \emph{topological invariant} of $\Delta_\epsilon$, called the \emph{$k-$th Betti number} of the topological space $\Delta_\epsilon$ and denoted by $b_k^\epsilon$. By varying $\epsilon$, we can make a tablet of the change of the $k-$th Betti number, which is called the tablet of \emph{persistence barcode}. The tablet consists of some topological information of the data cloud as well as some \emph{metric} information, since the change of $\epsilon$ can represent the change of the distance in the metric space.
	
	Of course, the Betti numbers aren't the only topological invariants for the barcode. For example, the \emph{Euler character} $\chi^\epsilon=\sum_{u=0}^l (-1)^u b_u^\epsilon$, where $l$ is the maximal number for which the $l-$th chain group isn't $0$, is another topological invariant. The aim of this paper is to introduce some commutative algebraic tools to obtain more topological or algebraic invariants. Also notice that in reality, the complex above may be directly given from more geometric data like the chemical data and the data can be also from some geometrical objects like manifolds, and in these cases, the persistent process may need some changes.
	
		\subsection{Commutative algebraic objects linked to simplicial complexes}\label{section: algebraic objects linked to simplicial complexes}
		 In this part, we will introduce some commutative algebraic objects linked to simplicial complexes, namely the edge ideal and Stanley--Reisner ideal of a simplicial complex. Basic knowledge of commutative algebra is presented in the Appendix for those in need. Most of the objects introduced here are studied by scholars in combinatoric commutative algebra (c.f. \cite{miller2006combinatorial,Hibi2011monomial,Nguyen2024,Seyed2025,Ha2024,ta2024,hien2024,fr2024,ibar2024}), but their main focus is on the algebraic or combinatoric aspects and does little with TDA. In the next section we will discuss  how they can be applied to TDA and give a \emph{persistent} version of these algebraic objects. 
		 
		 \subsubsection{Edge Ideals}
		 Let \( G = (V, E) \) be an undirected simple graph with vertex set \( V = \{1,\dots,n\} \), and \( \mathbb{S}_n = \mathbb{F}[x_1, \dots, x_n] \). The \emph{edge ideal} of \( G \), denoted by \( I_{\rm edge}(G)\), is the ideal defined as  
		 \begin{equation*}
		 	I_{\rm edge}(G) := \langle x_i x_j \mid \{i, j\} \in E \rangle_{\mathbb{S}_n}.
		 \end{equation*}
		 
		 We can also generalize the edge ideal to a simplicial complex. Specifically, let $\Delta$ be a simplicial complex over the vertex set $\{1,\dots,n\}$, the edge ideal of $\Delta$, denoted as $I_{\rm edge}(\Delta)$, is defined by 
		 \begin{equation*}
		 	I_{\rm edge}(\Delta) := \langle x_\sigma \mid \sigma \in \Delta, \ \dim(\sigma) = 1 \rangle_{\mathbb{S}_n}.    
		 \end{equation*}
		 	\subsubsection{Stanley--Reisner Ideals}
		 
		 As for any simplicial complex $\Delta$ over $\{1,\dots,n\}$, we define the \emph{Stanley–Reisner ideal (over $\mathbb{S}_n$)} $$I_\Delta:=\langle x_\sigma|\sigma\subset \{1,2,\dots,n\}, \sigma\neq \emptyset ,\sigma\not\in\Delta\rangle_{\mathbb{S}_n}$$ and the \emph{face ring or Stanley–-Reisner ring} is defined as the quotient ring $R_\Delta:=\mathbb{S}_n/I_\Delta$. Note that for a simplex $\Delta$ over $\{1,2,\dots,n\}$, since $\{\sigma|\sigma\subset \{1,2,\dots,n\}, \sigma\neq \emptyset ,\sigma\not\in\Delta\}=\emptyset$, then by logical convention, $I_{\Delta}=0$.
		
		 \begin{remark}\label{remark on ideals and complexes}
		 	\begin{enumerate}
		 		\item By definition, given the vertex set $\{1,\dots,n\}$, the edge ideal of a graph is determined by the graph and vice versa. However, since many simplicial complexes share the same set of one-dimensional faces, there doesn't exist a one-to-one correspondence between simplicial complexes and a special kind of ideals via the definition of the edge ideal.
		 		\item By definition, given the vertex set $\{1,\dots,n\}$, the Stanley–Reisner ideal is determined by the simplicial complex, so is the face ring. Conversely, any \emph{radical} monomial ideal $I\subset\mathbb{S}_n$ can determine a simplicial complex over $\{1,\dots,n\}$ whose Stanley–Reisner ideal in $\mathbb{S}_n$ is $I$ \cite[Theorem 1.7,Chapter 1]{miller2006combinatorial}. That's why Stanley–Reisner ideals are important in studying monomial ideals.
		 	\end{enumerate}
		 \end{remark}
		 For our purpose, the topology of the simplicial complex is vital, so the following rigidity theorem of Stanley–Reisner ideals are somehow important for us:
		 \begin{theorem}\cite[Theorem]{georgian1996invariance}
		 	Let $\Delta$ and $\Delta'$ be two abstract simplicial complexes defined on the vertex sets $ \{1,\dots,n\}$ and $\{1,\dots,m\}$ respectively. Suppose $\mathbb{F}[\Delta]$ and $\mathbb{F}[\Delta']$ are isomorphic as $\mathbb{F}-$algebra. Then $m = n$ and there exists a permutation $\phi \in S_n$ such that $\phi$ induces an isomorphism between $\Delta$ and $\Delta'$.
		 \end{theorem}
		\section{Persistent ideals linked to persistent simplicial complex}
		In this section, we will first discuss some properties of edge ideals and Stanley–Reisner ideals. Then we will give the definition of \emph{the persitent ideals linked to persistent simplicial complex} and use these as the new barcodes.  Based on the properties we have discussed, we will prove that these new barcodes will cover the information of the traditional ones. Most inportantly, the method we propose here is not just for edge ideals and Stanley–Reisner ideals, but can also be applied to other ideals linked to simplicial complex.

		\subsection{Properties of the edge ideals and Stanley–Reisner ideals}
		Note that both edge ideal and Stanley–Reisner ideal are monomial ideals, so we can consider the \emph{geometric or topological} meanings of concepts and properties mentioned in the Lemma~\ref{minimal basis of monomial ideals},\ref{radical monomial ideal},\ref{primary decomposition of a radical monomial ideal} and Corollary~\ref{prime decomposition of a radical monomial ideal} in the appendix.
		\begin{proposition}\label{minimal basis of edge ideal and the Stanley–Reisner ideals}
		\begin{enumerate}
		\item For the edge ideal $I_{edge}(\Delta)$ of a simplicial complex $\Delta$, the set of monomials $\{x_\sigma| \sigma\in\Delta,\#\sigma=1\}$ in the definition is a minimal basis of $I_{edge}(\Delta)$.
		\item For the Stanley–Reisner ideal $I_\Delta$ of a simplicial complex $\Delta$,  the set $\{x_\sigma|\sigma\neq \emptyset , \sigma\not\in \Delta, F\sigma\subset \Delta\}$, where $F\sigma$ is the set of all proper nonempty subset of $\sigma$ ,is a minimal basis.
		\end{enumerate}	
		\end{proposition}
		\begin{proof}
			The first assertion is directly from the definition of edge ideal and Lemma~\ref{minimal basis of monomial ideals}.
			
			For the second assertion, denote $W:=\{x_\sigma|\sigma\neq \emptyset , \sigma\not\in \Delta, F\sigma\subset \Delta\}$. For any $\sigma\neq \emptyset$ and $\sigma\not\in \Delta$, if $F\sigma\not\subset \Delta$, we can choose $\sigma'\subset \sigma$ such that $F\sigma'\subset\Delta$ by induction. Note that $x_{\sigma'}|x_\sigma$. Combining this with the definition of Stanley–Reisner ideal and Lemma~\ref{Criterion for members of monomial ideals}, we can conclude that $W$ is a set of generators of $I_\Delta$. The minimality of $W$ is easy to see from the definition of $W$ and $I_\Delta$. 
		\end{proof}
		\begin{proposition}\label{SR ideals and edge ideals are radical}
			The Stanley–Reisner ideal $I_\Delta$ and the edge ideal $I_{edge}(\Delta)$ for a simplicial complex $\Delta$ over $\{1,\dots,n\}$ are always radical.
		\end{proposition}
		\begin{proof}
			Just notice that $x_\sigma$ for $\sigma\subset \{1,2,\dots,n\}$ is square-free and the assertion holds according to Lemma~\ref{radical monomial ideal}.
		\end{proof}
		\begin{corollary}\label{decomposation of ideals of complex}
			The Stanley–Reisner ideal $I_\Delta$ and the edge ideal $I_{\rm edge}(\Delta)$ for a simplicial complex $\Delta$ over $\{1,\dots,n\}$ can be prime decomposed as the intersection of linear prime ideals.
		\end{corollary}
		\begin{proof}
			Directly from the Proposition \ref{SR ideals and edge ideals are radical} and the corollary \ref{prime decomposition of a radical monomial ideal}.
		\end{proof}
		Let's see what the prime decomposition and the associated primes really mean in the case of $I_\Delta$ and $I_{edge}(\Delta)$(or $I_{edge}(G)$ for an undirected simple graph $G$). Recall that a \emph{vertex cover} $W$ of a graph $G=(V,E)$ is a subset of $V$ such that for each $e=\{i,j\}\in E$, either $i\in W$ or $j\in W$. And a vertex cover $W$ is said to be \emph{minimal} if any proper subset of $W$ is not a vertex cover of $G$.
		\begin{lemma}\label{relations between the ideals and simplex/graph}
			\begin{enumerate}
				\item For two graphs $G_1,G_2$ with vertex set $\{1,\dots,n\}$, $G_1$is a subgraph of $ G_2$ $\Longleftrightarrow I_{edge}(G_1)\subset I_{edge}(G_2)$. 
				\item For two simplicial complexes $\Delta,\Delta'$ over $\{1,\dots,n\}$, $\Delta'\subset \Delta\Longleftrightarrow I_\Delta \subset I_{\Delta'}$
				\item Let $W\subset \{1,\dots,n\}$ and $W^c=\{i\in\{1,\dots,n\}|i\not\in\mathcal{W}\}$. Then $\Delta$ is the simplex over $W$ if and only if its Stanley–Reisner ideal $I_{\Delta}$ in $\mathbb{S}_n$ is the linear prime ideal generated by $\{x_i|i\in W^c\}$.
			\end{enumerate}
		\end{lemma}
		\begin{proof}
			Simply from the definitions.
		\end{proof}
		\begin{proposition}\label{topological meaning of the associated prime}
			\begin{enumerate}
				\item\cite[Theorem 34\& Corollary 35, Part 2]{bigatti2013} For an undirected simple graph $G=(V,E)$ with $V=\{1,\dots,n\}$ and $W\subset V$, 
				$I_{edge}(G)$ is a subset of $P_{W}:=\langle x_i|i\in W\rangle_{\mathbb{S}_n}$ 
				if and only if $W$ is a vertex cover of $G$.
				In particular, $P_{W}$ belongs to $\Ass(I_{edge}(G))$ if and only if $W$ is a minimal vertex cover of $G$.
				\item \cite[Proposition 4.11]{francisco2014survey} For a simplicial complex $\Delta$ on $\{1,\dots,n\}$ and $W\subset\{1,\dots,n\}$, 
				$I_{\Delta}$ is a subset of $P_{W}:=\langle x_i|i\in W\rangle_{\mathbb{S}_n}$ 
				if and only if the full subcomplex on $W^c:=\{j\in \{1,\dots,n\}|j\not\in W\}$ of $\Delta$ is a simplex.
				In particular, $P_{W}$ belongs to $\Ass(I(\Delta))$ if and only if the full subcomplex of $\Delta$ on $W^c$  is a maximal simplex in the sense that for $U\supset W^c$, the full subcomplex on $U$ of $\Delta$ is not a simplex. 
			\end{enumerate}
		\end{proposition}
		\subsection{Persistent ideals and their properties}
		Now we consider the "persistent version" of decomposition of ideals:
		\begin{theorem}[Persistent version of the decomposition of ideals]\label{Proposition: lemma for the birth-death information}
			
			Let $I_1 \supseteq I_2 \supseteq I_3 \supseteq \cdots$ be a descending chain of ideals of a ring $R$. Suppose each $I_i$ is a prime decomposable ideal.  If $P \in \Ass(I_i)$ and $P \notin \Ass(I_{i+1})$, then $P \notin \Ass(I_j)$ whenever $j \geq i+1$. 
		\end{theorem}
		\begin{proof}
			Because $P \in \Ass(I_i)$, $I_{j} \subset I_{i+1} \subset I_i \subset P$ for $j\geq i+1$. On the other hand, since $P\not\in\Ass(I_{i+1})$, by proposition~\ref{Proposition: Associated primes are precisely the minimal elements}, there exists $P'\in \Ass(I_{i+1})$ such that $I_j\subset I_{i+1}\subset P'\subset P$ for $j\geq i+1$ and $P'\neq P$. So again by proposition~\ref{Proposition: Associated primes are precisely the minimal elements}, $P\not\in\Ass(P_{j})$ for all $j\geq i+1$.
		\end{proof}
		By the Lemma~\ref{relations between the ideals and simplex/graph}, we can get that for an increasing filtration of persistent simplicial complex $\Delta_{\epsilon}$  when $\epsilon$ increases,  the Stanley–Reisner ideals forms a descending chain of ideals and the edge ideals forms an ascending chain of ideals, both with finite length.  
		Note that Proposition~\ref{Proposition: lemma for the birth-death information} is also valid for an ascending chain with \emph{finite length}, just by relabeling the ideals. Thus we can have the following definition:
		\begin{definition}\label{definition of persistent ideals}
		For a persistent (reduced) simplicial complex $\Delta_\epsilon$, we define \emph{the persistent ideals} for $\Delta_{\epsilon}$ as \emph{the set of associated primes} $\Ass(I_{\Delta_\epsilon})$ or $\Ass(I_{edge}(\Delta_\epsilon))$ and use them as new persistence barcodes.
		\end{definition}
		\begin{remark}
			\begin{enumerate}
				\item From 	Proposition~\ref{topological meaning of the associated prime}, the persistent ideals have good topological illustration: the birth and death of the associated primes of the Stanley–Reisner ideal reflect the birth and death (more precisely, the loss of maximality) of the maximal simplex in the simplicial complex; the birth and death of the associated primes of the edge ideal inform us the birth and death of the minimal vertex cover of the graph.
				\item In commutative algebra, Proposition~\ref{Proposition: lemma for the birth-death information} states that the \emph{height} of the ideals in a descending chain is decreasing. So this numerical quantity can be used as the Betti numbers in the persistent process.
				\item Note that this persistent method is not just valid for the two ideals defined above, but it's also valid for other ideals like \emph{the cover ideals} defined in \cite{bigatti2013} or those in \cite{Hibi2011monomial, Hibi2018binomial}. We may get different topological data from different ideals.
			\end{enumerate}
		\end{remark}
		
			The persistent ideals ensure us to construct new barcodes instead of the usual one like the Betti numbers. Whether the new one can cover the information of the old ones is a natural question. The following example gives us some inspiration to confirm the question.
			\begin{example}
				Figure~\ref{fig:enter-label} illustrates an example of the new persistence barcodes and the old ones. The point cloud $\mathcal{D}\subset \mathbb{R}^3$ here is defined  by the distant matrix shown in the left part of the figure. In the figure, we can easily see that all the lifespan information of the tranditional barcodes of dimension 0 can be read in the new barcodes defined by the associate primes of the Stanley--Reisner ideals.
			\end{example}
			\begin{figure}
				\centering
				\includegraphics[width=\linewidth]{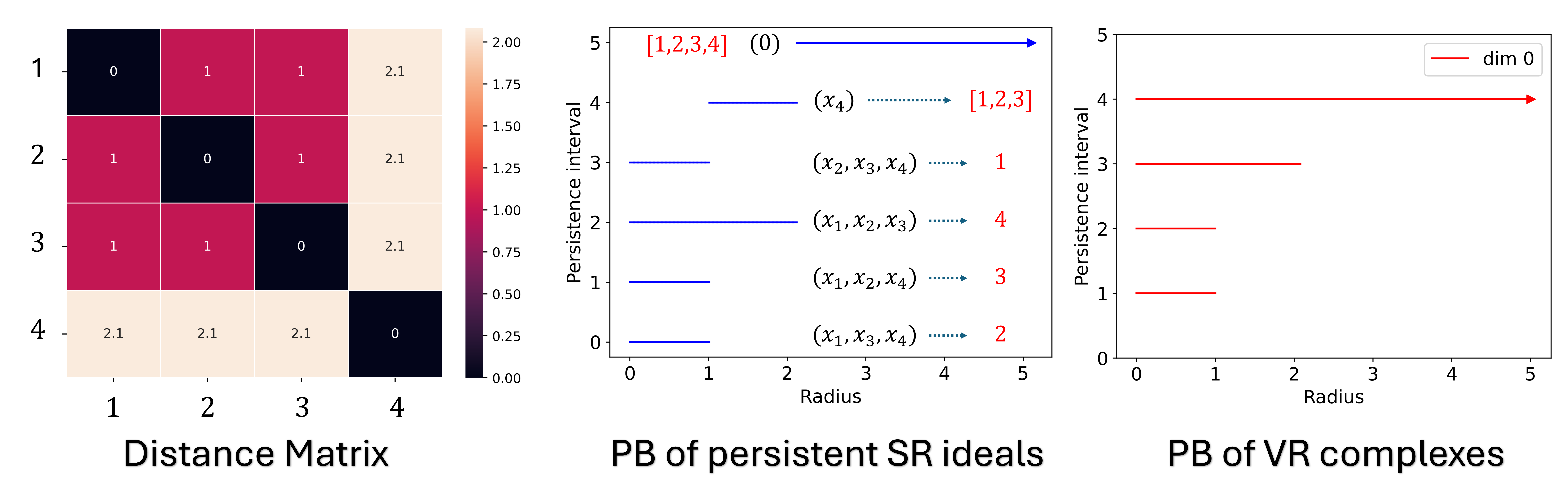}
				\caption{Illustration of an example distance matrix for a point cloud consisting of four points indexed by 1, 2, 3, and 4, along with the corresponding persistence barcodes (PBs) induced by the persistent Stanley--Reisner (SR) ideals and conventional persistent homology. Both PBs are derived from the same Vietoris--Rips (VR) filtration constructed over the point cloud. Each persistence interval in the PB of the persistent SR ideals corresponds to a monomial ideal generated by a subset of \( \{ x_i \mid i = 1, 2, 3, 4 \} \), which is annotated by black text next to the corresponding interval. Equivalently, as indicated by the red annotations, each monomial ideal corresponds to a maximal simplex in the simplicial complex associated with a fixed radius in the VR filtration, encoding the lifespan information of maximal simplices of various dimensions within the VR filtration.}
				\label{fig:enter-label}
			\end{figure}

			 Now we give a strict proof of the following theorem to affirm the answer for the questions above in the case of the persistent associated primes of the Stanley–Reisner ideal and the Betti numbers in traditional TDA in a more general setting. 
			
			First we make some notations and assumption. Let $\{\Delta_\epsilon\}$ be persistent simplicial complexes over $\{1,\dots,n\}$. Denote $b_{i}(\epsilon)$ be the i-th Betti number of $\Delta_\epsilon$ and view it as a function of $\epsilon$ from $\mathbb{R}$ to $\mathbb{N}$. We assume that $b_{i}(\epsilon+):=\lim_{\rho\to\epsilon+}b_i(\rho)$ and $b_i(\epsilon-):=\lim_{\rho\to\epsilon-}b_{i}(\rho)$ exist for any $\epsilon$. 
			
			For a  prime ideal $P$, we define $P(\epsilon)$ to be the function of $\epsilon$ from $\mathbb{R}$ to $\{0,1\}$ such that $P(\epsilon)=1$ if $P\in \Ass(I_{\Delta_{\epsilon}})$ and $P(\epsilon)=0$ if $P\not\in \Ass(I_{\Delta_{\epsilon}})$. We also assume that $P(\epsilon+):=\lim_{\rho\to\epsilon+}P(\rho)$ and $P(\epsilon-):=\lim_{\rho\to\epsilon-}P(\rho)$ exist for any $\epsilon$.  
			
			
			\begin{remark}\label{remark:continuity assumption}
				Since $\mathbb{N}$ has the discrete topology, then the assumption that $b_{i}(\epsilon+)$ exists is in fact equivalent to the assumption that for every $\epsilon$, there exists $\delta>0$, such that $b_{i}(\epsilon+\rho)$ is constant for $\rho\in (0,\delta)$. This is consistent with the intuition of the traditional barcode, so in the traditional persistent complexes (which is a filtration), the assumption holds. Similar arguments hold for the assumption that $b_i(\epsilon-),P(\epsilon+),P(\epsilon-)$ exist because $\{0,1\}$ also has discrete topology.
			\end{remark}

			\begin{theorem}[Persistent ideals preserve informations in the traditional ones]\label{Persistent affine variety covers persistent Betti number}
				
				With the notation and assumption above, for any $i_0$ and $\epsilon_0$, if $b_{i_0}(\epsilon_0+)-b_{i_0}(\epsilon_0-)\neq 0$, then there exists a linear prime ideal $P$ of the form $P=\langle x_{i_1},x_{i_2},\dots,x_{i_k}\rangle_{\mathbb{S}_n}$ with $1\leq i_1<i_2<\dots<i_k\leq n$, such that $P(\epsilon_0+)-P(\epsilon_0-)\neq 0$. 
				
			\end{theorem}
			\begin{proof}
				
				Assume that for $i_0$ and $\epsilon_0$, $b_{i_0}(\epsilon_0+)-b_{i_0}(\epsilon_0-)\neq 0$. We'll prove that $P(\epsilon_0+)-P(\epsilon_0-)\neq 0$ for some linear prime ideal $P$ by contradiction. Denote $\mathcal{LP}$ as the set of linear prime ideal $P$ of the form $P=\langle x_{i_1},\dots,x_{i_k}\rangle_{\mathbb{S}_n}$ with $1\leq i_1<i_2<\dots<i_k\leq n$.
				
				Suppose by contrary that $P(\epsilon_0+)-P(\epsilon_0-)=0$ for $P\in \mathcal{LP}$ . Note that $\#\mathcal{LP}$ is finite. So by the remark ~\ref{remark:continuity assumption}, there exists $\delta>0$, such that for any $P\in\mathcal{LP}$ and all $\rho\in (0,\delta)$,
				\begin{equation}\label{eq:local constant of associated prime}
						P(\epsilon_0+\rho)=P(\epsilon_0+)\quad
						P(\epsilon_0-\rho)=P(\epsilon_0-)
				\end{equation}
				Let $J:=\bigcap_{P\in \mathcal{LP},P(\epsilon_0+)=1}P$ and the simplicial complex $\Delta'$ whose Stanley–Reisner ideal is $J$, by Remark~\ref{remark on ideals and complexes}. Denote $b_{i_0}(\Delta')$ the $i_0$-th Betti number of $\Delta'$.
				Then by definition and \eqref{eq:local constant of associated prime}, we get that for any $\rho\in (0,\delta)$,  $I_{\Delta_{\epsilon_0-\rho}}=J= I_{\Delta_{\epsilon_0+\rho}}$. 
				Therefore by definition,  $b_{i_0}(\epsilon_0+)=b_{i_0}(\Delta')=b_{i_0}(\epsilon_0-)$, which leads to a contradiction.  
			\end{proof}
			Combining  Theorem~\ref{Persistent affine variety covers persistent Betti number} and  Lemma~\ref{relations between the ideals and simplex/graph}, we give the result for the Vietoris-Rips complex on the information provided by the new persistence barcodes like the traditional one in terms of the distance matrix, but it is stronger than the traditional one.  
			\begin{theorem}[Information from the new barcode: case of the distance matrix]\label{Persistent affine variety covers persistent Betti number: dicrete version}
				
				Given $\mathcal{D}:=\{x_1,\dots,x_n\}$ contained in a metric space $(X,d)$, denote the distance matrix of $\mathcal{D}$ by $H:=(h_{ij})$, $h_{ij}=d(x_i,x_j)$. Denote $\mathcal{LP}$ as the set of linear prime ideals $P$ of the form $P=\langle x_{i_1},\dots,x_{i_k}\rangle_{\mathbb{S}_n}$ with $1\leq i_1<i_2<\dots<i_k\leq n$. 
				 For the Vietoris-Rips complex, all the birth and death information of $\frac{h_{ij}}{2}(\forall i\neq j)$ can be reflected in the birth and death information of one of the persistence barcodes of elements in $\mathcal{LP}$ as for the Stanley–Reisner ideal. 
			\end{theorem}
			\begin{proof}
				Since $H$ is a symmetric matrix, we only need to consider the case where $i<j$. Let $\delta_1:=\frac{1}{4}\min_{i<j,k<l, i\neq k, j\neq l} |h_{ij}-h_{kl}|$, $\delta_2:=\frac{1}{4}\min_{i<j} h_{ij}$, $\delta:=\min\{\delta_1,\delta_2\}$. According to the definition of Vietoris-Rips complex: for any $i<j$, let $\epsilon=\frac{h_{ij}}{2}$, then for $\rho\in (\epsilon-\delta,\epsilon)$, $\{i,j\}\not \in \Delta^{\rho}$; and for $\eta\in (\epsilon,\epsilon+\delta)$, $\{i,j\}\in \Delta^{\eta}$. So by Proposition~\ref{topological meaning of the associated prime} and Theorem~\ref{Persistent affine variety covers persistent Betti number}, there exists a linear prime $P\in \mathcal{LP}$ such that $P(\rho)\neq P(\eta)$. Let $\eta\mapsto \epsilon+, \rho\mapsto \epsilon-$, we get $P_{\mathfrak{p}}(\epsilon+)\neq P_{\mathfrak{p}}(\epsilon-)$. Therefore $\epsilon$ is the time of birth or death of the barcode of $P$.
			\end{proof}
			\begin{remark}
				\begin{enumerate}
					\item The converse of the Theorem~\ref{Persistent affine variety covers persistent Betti number} and Theorem~\ref{Persistent affine variety covers persistent Betti number: dicrete version} are NOT true as the following Example~\ref{counterexample: the old barcode doesn't covers the new ones} shows.
					Actually, the set of the barcodes of persistent linear prime ideals is the \emph{complete} set of algebraic invariants for the persistent complexes, meaning that the set of persistent linear prime ideals  uniquely determines and is determined by the persistent complexes, by Remark~\ref{remark on ideals and complexes} and the uniqueness of the prime decomposition~\ref{Strong uniqueness of prime decomposition}. But all the common topological quantities like Betti numbers are \emph{partial} set of algebraic invariants for the persistent complexes in general.
					Thus the assumption made above is just for convenience to prove but not in essence. And the role of the Betti numbers can also be replaced by any other common topological invariants and Theorem~~\ref{Persistent affine variety covers persistent Betti number} and Theorem~~\ref{Persistent affine variety covers persistent Betti number: dicrete version} are also true. 
					\item Other kinds of persistent ideals may also cover the information from the old PB, but we'll not discuss them here.  
				\end{enumerate}
			\end{remark}
		\begin{example}\label{counterexample: the old barcode doesn't covers the new ones}[Counterexample for the invserse of Theorem~\ref{Persistent affine variety covers persistent Betti number} and Theorem~\ref{Persistent affine variety covers persistent Betti number: dicrete version}]
			Consider the data cloud $\mathcal{D}\subset \mathbb{R}^2$ consists of three points $(0,0),(2,0),(0,2)$ named by $1,2,3$ respectively. The distance matrix $H$ of $\mathcal{D}$ is defined as follows:
			$$H:=\begin{pmatrix}
				0&2&2\\
				2&0&2\sqrt{2}\\
				2&2\sqrt{2}&0
			\end{pmatrix}$$ 
			The persistent simplicial complex $\Delta^{\epsilon}$ can be shown as follows:
			\[
			\begin{tikzpicture}
					\filldraw [black] (-5,0) circle (2pt)node[anchor=east]{1};
				  \filldraw [black] (-3,0) circle (2pt)node[anchor=west]{2};
                  \filldraw [black] (-5,2) circle (2pt)node[anchor=east]{3};
                  \node at (-4,-1){\text{Complex for} $\epsilon\in [0,2)$};
                  \draw [->] (-2,1)--(-1,1);
				\filldraw [black] (0,0) circle (2pt)node[anchor=east]{1};
				\filldraw [black] (2,0) circle (2pt)node[anchor=west]{2};
				\filldraw [black] (0,2) circle (2pt)node[anchor=east]{3};
				\draw[blue,thick] (0,0)--(0,2);
		    	\draw[blue,thick] (0,0)--(2,0);
				\node at (1,-1){\text{Complex for} $\epsilon\in [2,\sqrt{2})$};
				\draw [->] (3,1)--(4,1);
				\filldraw [black] (5,0) circle (2pt)node[anchor=east]{1};
				\filldraw [black] (7,0) circle (2pt)node[anchor=west]{2};
				\filldraw [black] (5,2) circle (2pt)node[anchor=east]{3};
				\draw[blue,thick] (5,0)--(5,2);
				\draw[blue,thick] (5,0)--(7,0);
				\draw[blue,thick] (5,2)--(7,0);
				\filldraw[gray] (5,0)--(5,2)--(7,0)--cycle;
				\node at (6,-1){\text{Complex for} $\epsilon\in [2,\sqrt{2})$};
			\end{tikzpicture}
			\]
			Note that the last two figures are both \emph{contractible} and connected, so for any $\epsilon\geq 2$, the \emph{topological type} of $\Delta^{\epsilon}$ is always the same. 
			 So the quantity $\frac{h_{2,3}}{2}=\sqrt{2}$ will not be presented in the tranditional barcodes of Betti numbers and other topological invariants. But if we consider the persistent ideals as barcodes, we can get $\frac{h_{2,3}}{2}=\sqrt{2}$ as the death time of the linear prime ideal $\langle x_2\rangle_{\mathbb{S}_3}$ or $\langle x_3 \rangle_{\mathbb{S}_3}$. This is because for $\epsilon\in[2,\sqrt{2})$, $I_{\Delta^{\epsilon}}=\langle x_2x_3\rangle_{\mathbb{S}_3}=\langle x_2\rangle_{\mathbb{S}_3}\bigcap\langle x_3\rangle_{\mathbb{S}_3}$, but for $\epsilon\in[\sqrt{2},+\infty)$, $I_{\epsilon}=0$. 
		\end{example}	
		Finally, for a graph $G$, let's see the relation between the Stanley–Reisner ideal of the \emph{clique complex} $\Delta$ of $G$ and the edge ideal of the complemet graph of the graph $G$.
		\begin{proposition}[Edge ideal vs Stanley–Reisner ideal, case of clique complex]\label{relation between the edge ideal and the Stanley–Reisner ideal}
			Let $G=(V,E)$ be a graph with vertices $\{1,\dots,n\}$ and $\Delta$ the clique complex of $G$. Let $\tilde{G}$ be the complement graph of $G$. Then $I_{\Delta}=I_{edge}(\tilde{G})$. 
		\end{proposition}
		\begin{proof}
			By proposition~\ref{minimal basis of edge ideal and the Stanley–Reisner ideals}, $I_{\Delta}$ has a minimal set of generators $\mathcal{B}:=\{x_\sigma|\sigma\subset \{1,\dots,n\}, \sigma\not\in\Delta, \sigma\neq \emptyset, F\sigma\subset \Delta\}$, where $F\sigma$ is the set of all proper subfaces of $\sigma$. For $\sigma\subset \{1,\dots,n\}$ and $\# \sigma\geq 3$, if $F\sigma\subset \Delta$, then surely all the $1-$dimensional subfaces of $\sigma$ belong to $\Delta$ and by definition of clique complex, $\sigma\in \Delta$.  So $\mathcal{B}$ contains $x_\sigma$ with $\#\sigma\leq 2$. Since the vertex set $V$ is the full set $\{1,\dots,n\}$, then by definition $\forall i\in\{1,\dots,n\}$, $\{i\}\in\Delta$. Thus $\mathcal{B}=\{x_\sigma|\sigma\not\in\Delta,\#\sigma=2\}$. Therefore again by definition $I_{\Delta}=\langle \mathcal{B}\rangle_{\mathbb{S}_n}=I_{edge}(\tilde{G})$.  
		\end{proof}
		\begin{remark}
			\begin{enumerate}
			\item It needs to emphasize that Proposition~\ref{relation between the edge ideal and the Stanley–Reisner ideal} ONLY holds for the \emph{clique complex} of a graph, not for other complex of a graph, see the following example~\ref{eg: relation between edge ideals and SR ideals}.
			\item Since the Vietoris-Rips complex is the clique complex of its underlying graph, then by Proposition~\ref{relation between the edge ideal and the Stanley–Reisner ideal},  Theorem~\ref{Persistent affine variety covers persistent Betti number} and Theorem~\ref{Persistent affine variety covers persistent Betti number: dicrete version}, we can get that \emph{in the case of the Vietoris-Rips complex, the persistent ideals of edge ideals can also covers the information provided by the traditional persistence barcodes.}
			\end{enumerate} 
		\end{remark}
		\begin{example}\label{eg: relation between edge ideals and SR ideals}
			Let $G=(V,E)$ be the graph with 
			$$V:=\{1,2,3,4\},\quad E:=\{\{1,2\},\{1,3\},\{2,3\},\{3,4\}\}.$$
			Its complement graph $\tilde{G}:=(V,\tilde{E})$, where $\tilde{E}:=\{(i,j)|(i,j)\not\in E\}=\{\{2,4\},\{1,4\}\}$. Then the clique complexes $\Delta$ and $\tilde{\Delta}$ of $G$ and $\tilde{G}$ are defined by
			\begin{equation*}
				\begin{aligned}
					\Delta&=\{\{1\},\{2\},\{3\},\{4\},\{3,4\},\{1,2\},\{1,3\},\{2,3\},\{1,2,3\}\},\\
					\bar{\Delta}&=\{\{1\},\{2\},\{3\},\{4\},\{1,4\},\{2,4\}\}.
				\end{aligned}
			\end{equation*}
			
			It's easy to calculate that 
			\begin{equation*}
			\begin{aligned}
				I_\Delta&=\langle x_1x_4,x_2x_4,x_1x_2x_4,x_1x_3x_4,x_2x_3x_4,x_1x_2x_3x_4\rangle_{\mathbb{S}_4}=\langle x_1x_4,x_2x_4\rangle_{\mathbb{S}_4}\\
				I_{\bar{\Delta}}&=\langle x_1x_2,x_1x_3,x_2x_3,x_3x_4,x_{\sigma}|\sigma\subset \{1,2,3,4\},\#\sigma\geq 3\rangle_{\mathbb{S}_4}=\langle x_1x_2,x_1x_3,x_2x_3,x_3x_4\rangle_{\mathbb{S}_4}\\
				I_{edge}(\Delta)&=\langle x_1x_2,x_1x_3,x_2x_3,x_3x_4\rangle_{\mathbb{S}_4}\\
				I_{edge}(\tilde{\Delta})&=\langle x_1x_4,x_2x_4\rangle_{\mathbb{S}_4}.
			\end{aligned}
			\end{equation*}
			So the conclusion of the proposition~\ref{relation between the edge ideal and the Stanley–Reisner ideal} holds. 
			
			But if we consider the complex $\Delta'$ defined by $G$  is $$\Delta'=\{\{1\},\{2\},\{3\},\{4\},\{3,4\},\{1,2\},\{1,3\},\{2,3\}\},$$
			then $$I_\Delta=\langle x_1x_4,x_2x_4,x_1x_2x_4,x_1x_3x_4,x_2x_3x_4,x_1x_2x_3x_4,x_1x_2x_3\rangle_{\mathbb{S}_4}=\langle x_1x_4,x_2x_4,x_1x_2x_3\rangle_{\mathbb{S}_4}$$ 
			and it's not equal to $I_{edge}(\tilde{\Delta})$.
			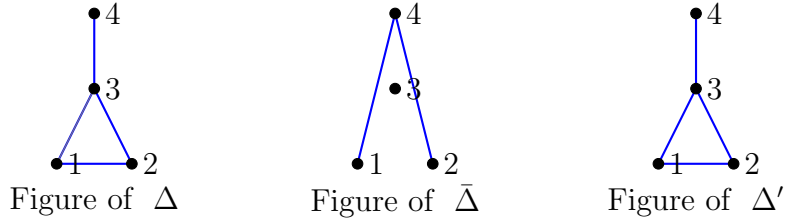
\begin{figure}[!htbp]
				\centering
					\begin{tikzpicture}
					\draw[blue,thick] (-4,0.5)--(-3.5,-0.5);
					\draw[blue,thick] (-4,0.5)--(-4,1.5);
					\draw[blue,thick] (-4,0.5)--(-4.5,-0.5);
					\draw[blue,thick] (-4.5,-0.5)--(-3.5,-0.5);
					\filldraw[gray] (-4,0.5)--(-4.5,-0.5)--(-4.5,-0.5)--cycle;
					\filldraw [black] (-4,0.5) circle (2pt)node[anchor=west]{3};
					\filldraw [black] (-4.5,-0.5) circle (2pt)node[anchor=west]{1};
					\filldraw [black] (-3.5,-0.5) circle (2pt)node[anchor=west]{2};
					\filldraw [black] (-4,1.5) circle (2pt)node[anchor=west]{4};
					\node at (-4,-1) {\text{Figure of } $\Delta$};
					\draw[blue,thick] (0.5,-0.5)--(0,1.5);
					\draw[blue,thick] (0,1.5)--(-0.5,-0.5);
					\filldraw [black] (0,0.5) circle (2pt)node[anchor=west]{3};
					\filldraw [black] (-0.5,-0.5) circle (2pt)node[anchor=west]{1};
					\filldraw [black] (0.5,-0.5) circle (2pt)node[anchor=west]{2};
					\filldraw [black] (0,1.5) circle (2pt)node[anchor=west]{4};
					\node at (0,-1) {\text{Figure of } $\bar{\Delta}$};
					\draw[blue,thick] (4,0.5)--(4,1.5);
					\draw[blue,thick] (4,0.5)--(3.5,-0.5)--(4.5,-0.5)--cycle;
					\filldraw [black] (4,0.5) circle (2pt)node[anchor=west]{3};
					\filldraw [black] (3.5,-0.5) circle (2pt)node[anchor=west]{1};
					\filldraw [black] (4.5,-0.5) circle (2pt)node[anchor=west]{2};
					\filldraw [black] (4,1.5) circle (2pt)node[anchor=west]{4};
					\node at (4,-1) {\text{Figure of } $\Delta'$};
				\end{tikzpicture}%
				\caption{Figures for the simplicial complexes in Example~\ref{eg: relation between edge ideals and SR ideals}}
			\end{figure}
		\end{example}	
		

		\section{Persistent chain complex of free modules}
		\subsection{Persistent labelled complex and persistent chain complex of free modules}
		
		\rm There is another way to associate a simplicial complex with algebraic objects,  by associating each vertex with some element in a ring $R$. It is discussed in many textbooks on commutative algebra like \cite{Hibi2011monomial}\cite{eisenbud2005syzygies} for studying the free resolution of an ideal.  It is also mentioned in the papers \cite{dawson1990homology, ren2018weighted} as the \emph{weighted complex} and in the references \cite[Chapter 8]{Schenck2022algebraic} and \cite{Carlsson2009} of the form \emph{multiparameter persistence} for TDA, but mainly about different properties from those discussed in our paper. 
		
		To be specific, let $R$ be a \emph{unique factorization domain (UFD)} and $\Delta$ is a simplicial complex over $\{1,\dots,n\}$. For each $i$, choose an element $m_i\in R$ with $m_i\neq 0$, which will be called the \emph{label element}, and denote $S=\{m_1\dots,m_n\}$.  Then we associate each face $\sigma\in\Delta$ the element $m_\sigma:=\Lcm\{m_i|i\in \sigma\}$(which are also called the label elements) where $\Lcm$ is the operator to get one of the \emph{least common multiplier}  of the given set of elements. This is well-defined because in UFD, the least common multiplier of any set of nonzero elements exists and uniquely defined by moduling a unit. Here we require that once we fix the choice of the least common multipliers, we will not change them by multiplying a unit in the following calculation.  A simplicial complex with label elements is called a \emph{labelled complex} and will be denoted by $(\Delta,S)$.
		
		  Given a labelled complex $(\Delta,S)$, we can get the following chain complex $\mathcal{C}(\Delta,S)$, called \emph{the associated chain complex of free $R-$modules} (for the label complex $(\Delta,S)$) :   
		\begin{equation}\label{associated chain complex}
		\mathcal{C}(\Delta,S):	\begin{tikzcd}
				0\arrow[r,"\tilde{\delta}_{n+1}"] &C_n(S) \arrow[r,"\tilde{\delta}_n"]& \cdots \arrow[r,"\tilde{\delta}_{2}"] & C_1(S) \arrow[r,"\tilde{\delta}_1"]& C_0(S) \arrow[r,"\tilde{\delta}_0"] &0
			\end{tikzcd}
		\end{equation}
		
		Here $C_k(S)$ is the free $R-$module with a basis of the $k-$dimensional faces of $\Delta$, that is, $$C_k(S)=\left\{ 
		\left. \sum_{\sigma\in \Delta,\# \sigma=k+1} p_\sigma \sigma \ \right| \ p_\sigma\in R \right\}.$$ The $R-$linear map $\tilde{\delta}_k:C_k(S)\longrightarrow C_{k-1}(S)$ is defined by 
		\begin{equation}\label{the definition of the boundary map of the associated complex}
			\tilde{\delta}_k(\sigma)=\sum_{u}(-1)^{u} \frac{m_\sigma}{m_{\sigma\setminus\{i_u\}}}\sigma\setminus\{i_u\}
		\end{equation}
		for $\sigma=\{i_1,\dots,i_{k+1}\}\in \Delta$ with $i_1<\dots<i_{k+1}$ and extended linearly to the whole $R-$module. Note that the coefficient $\frac{m_\sigma}{m_{\sigma\setminus\{i_u\}}}$ is indeed in $R$ by definition.
		
		Recall that the original chain complex $\mathcal{C}(\Delta)$ is defined as:
		\begin{equation}\label{original chain complex}
			\begin{tikzcd}
				0\arrow[r,"\delta_{n+1}"] &C_n \arrow[r,"\delta_n"]& \cdots \arrow[r,"\delta_{2}"] & C_1 \arrow[r,"\delta_1"]& C_0 \arrow[r,"\delta_0"] &0
			\end{tikzcd}
		\end{equation}
		with free $\mathbb{Z}-$modules $C_k=\{\sum_{\sigma\in \Delta,\# \sigma=k} q_\sigma\sigma|q_\sigma\in \mathbb{Z}\}$ and the $\mathbb{Z}-$linear map 
		\begin{equation}\label{the definition of the boundary map of the original complex}
			\delta_k(\sigma)=\sum_{u=1}^{k+1}(-1)^{u} \sigma\setminus\{i_u\}
		\end{equation}
		for $\sigma=\{i_1,\dots,i_{k+1}\}\in \Delta$ with $i_1<\dots<i_{k+1}$ and extended linearly to the whole $\mathbb{Z}-$module. If we choose $R=\mathbb{Z}$ and $m_i=1$ for all $i$, it's easy to see that the two complexes are equal to each other. So the associated complex just generalizes the original chain complex. 
		\begin{remark}\label{difference between the method of change of base ring/group and the method of associated complex}
			The construction we present here is similar to that constructed in the so-called \emph{universal coefficient theorem of homology/cohomology} in algebraic topology (cf. \cite[Chapter 3]{hatcher2002AT}):
			\[
			\begin{tikzcd}
				0\arrow[r,"\delta_{n+1}\otimes id_\mathcal{G}"] &C_n\otimes \mathcal{G} \arrow[r,"\delta_n\otimes id_\mathcal{G}"]& \cdots \arrow[r,"\delta_{2}\otimes id_\mathcal{G}"] & C_1\otimes \mathcal{G} \arrow[r,"\delta_1\otimes id_\mathcal{G}"]& C_0\otimes \mathcal{G} \arrow[r,"\delta_0\otimes id_\mathcal{G}"] &0
			\end{tikzcd}
			\]
			with $\mathcal{G}$ is an arbitrary group. But it's not the same because the map $\delta_{k}\otimes id_\mathcal{G}$ doesn't change the $\mathcal{G}-$coefficient while our definition of $\tilde{\delta}_k$ may "multiply" some $R-$coefficients.  Note that it may not easy to generalize our definition to $\mathcal{G}-$'module' for arbitrary group $\mathcal{G}$ since the least common multiplier can't be easy to generalize.  
		\end{remark}
		Similarly, we can also define the \emph{reduced labelled complex} and the \emph{reduced associated chain complex of free $R-$modules}. Here we \emph{always} require that $\emptyset$ is labeled by the unit element $1\in R$ and the associated $-1-$dimensional chain group $C_{-1}(\Delta)$ is the rank one free $R-$module with a $R-$basis $\mathcal{B}_{-1}:=\{\emptyset\}$. Naturally, the $R-$linear map $\tilde{\delta}_0: C_{0}(S)\longrightarrow C_{-1}(0)$ is defined by $\tilde{\delta}_0(\{i\})=-m_i\emptyset$ and extend linearly to the whole module. Note that in the original reduced chain complex, $\delta_0(\{i\})=-\emptyset$.
		
		\begin{definition}\label{definition: persistent labelled complex }
			Given a UFD $R$ and $S=\{m_1,\dots,m_n\}\subset R$. Suppose $\Delta_\epsilon$ is a persistent simplicial complex with vertex set $\{1,\dots,n\}$. We can 	define the \emph{persistent labelled complex}(labeled by $S$) as $(\Delta_\epsilon,S)$ and the corresponding \emph{persistent chain complex of free modules} as $\mathcal{C}(\Delta_\epsilon,S)$.  
		\end{definition}
		\subsection{Properties of persistent chain complex of free modules }
		First let's see a simple but crucial "formal" lemma for the relation between the two chain complexes: $\mathcal{C}(\Delta,S)$ as~\eqref{associated chain complex} and  $\mathcal{C}(\Delta)$ as~\eqref{original chain complex}. Since $\mathcal{B}_k:=\{A\in \Delta|\#A=k+1\}$ is a $\mathbb{Z}-$basis of $C_k$ and also an $R-$basis of $C_k(\mathcal{S})$, then we call this \emph{ the canonical basis} of the two complexes.
		\begin{lemma}\label{proposition:homotopy_matrix}
			Denote $\mathcal{B}_k=\{e_{k,1},\dots,e_{k,n_k}\}$ with $n_k=\dim(C_k)$ and let $D_k$ and $\tilde{D}_k$ be the matrices of the map $\delta_k$ and $\tilde{\delta}_k$ under $\mathcal{B}_k$ respectively, then we have the following \emph{formal} relation:
			\begin{equation}\label{homotopy_matrix}
				\tilde{D}_k=\diag(\frac{1}{m_{e_{k-1,1}}},\cdots,\frac{1}{m_{e_{k-1,n_{k-1}}}})D_k \diag(m_{e_{k,1}},\cdots,m_{e_{k,n_k}})
			\end{equation}
		\end{lemma}
		\begin{proof}
			Just combine the definition~\eqref{the definition of the boundary map of the associated complex} and~\eqref{the definition of the boundary map of the original complex}. 
		\end{proof}
		
		With this lemma, we can get the relation between the persistent chain complex of free modules and the tranditional persistent chain complex. 
		\begin{theorem}[Equivalence of $\mathcal{C}(\Delta)$ and $\mathcal{C}(\Delta,S)$, I]\label{equivalence of the two chain complexes via evaluation}
			
			Let $\mathfrak{m}$ be a maximal ideal of $R$ that contains none of the elements $m_i$ and $\mathbb{F}=R/\mathfrak{m}$ be the quotient ring (which is actually a field). Then the two following chain complexes are equivalent :
			\[
			\mathcal{C}(\Delta)\otimes_{\mathbb{Z}} \mathbb{F}:
			\begin{tikzcd}[row sep=1.5cm,column sep=1.5cm]
				\cdots \arrow[r, "\delta_{k+1}\otimes id_\mathbb{F}"]&C_k\otimes_{\mathbb{Z}} \mathbb{F} \arrow[r, "\delta_{k}\otimes id_\mathbb{F}"]& C_{k-1}\otimes_{\mathbb{Z}} \mathbb{F} \arrow[r, "\delta_{k-1}\otimes id_{\mathbb{F}}"] &\cdots 
			\end{tikzcd}
			\]
			\[ \mathcal{C}(\Delta,S)\otimes_{R} \mathbb{F}:
			\begin{tikzcd}[row sep=1.5cm,column sep=1.5cm]
				\cdots \arrow[r, "\tilde{\delta}_{k+1}\otimes id_\mathbb{F}"]&C_k(S)\otimes_{R} \mathbb{F} \arrow[r, "\tilde{\delta}_{k}\otimes id_{\mathbb{F}}"]& C_{k-1}(S)\otimes_{R} \mathbb{F} \arrow[r, "\tilde{\delta}_{k-1}\otimes id_\mathbb{F}"] &\cdots 
			\end{tikzcd}
			\]
			And hence the homology $\mathbb{F}-$modules of the two complexes are isomorphic.
			In particular, for the persistent version, we get that $\mathcal{C}(\Delta_\epsilon,S)\otimes_{R}\mathbb{F}$ is equivalent to $\mathcal{C}(\Delta_{\epsilon})\otimes_{\mathbb{Z}}\mathbb{F}$ for every $\epsilon$.
		\end{theorem}
		\begin{proof}
			Let the $\mathbb{F}-$linear map $\mathcal{A}_k: C_k\otimes_{\mathbb{Z}}\mathbb{F}\longrightarrow C_k(S)\otimes_{R}\mathbb{F}$ be defined by  the matrix 
			$\diag(\frac{1}{m_{e_{k,1}}},\dots,\frac{1}{m_{e_{k,n_{k}}}})$
			with respect to the canonical basis $\{e_{kl}\otimes 1\}$. Since $\mathfrak{m}$ contains none of $m_i$, then by definition of the least common multiplier, $m_A\not\in \mathfrak{m}$ for any $A\in\Delta$. Note that $\mathbb{F}=R/\mathfrak{m}$ is a field, therefore the map $\mathcal{A}_k$ is well-defined in the sense that $m_A$ is understood as the equivalence class of $m_A$ in the field $\mathbb{F}$ and $\frac{1}{m_A}$ is understood as the inverse of the equivalence class of $m_A$. By Lemma \ref{proposition:homotopy_matrix}, the following diagram commutes: 
			\[
			\begin{tikzcd}[row sep=1.5cm,column sep=1.5cm]
				C_k\otimes_{\mathbb{Z}}\mathbb{F}\arrow[d,"\mathcal{A}_k"] \arrow[r,"\delta_k\otimes id_\mathbb{F}"]& C_{k-1}\otimes_{\mathbb{Z}}\mathbb{F}\arrow[d,"\mathcal{A}_{k-1}"]\\
				C_k(S)\otimes_{R}\mathbb{F} \arrow[r,"\tilde{\delta}\otimes id_{\mathbb{F}}"]& C_{k-1}(S)\otimes_{R}\mathbb{F}
			\end{tikzcd}.
			\]
			Therefore, $\mathcal{A}:=\{\mathcal{A}_k\}$ is the equivalence map between these two complexes. 
		\end{proof}
		There is another equivalence of the chain complex via localization (which may be more natural in algebraic geometry):
		\begin{theorem}[Equivalence of $\mathcal{C}(\Delta)$ and $\mathcal{C}(\Delta,S)$, II]\label{equivalence of the two chain complexes via localization}
			
			Let $U$ be a multiplicative set of $R$ which contains each $m_i$. Denote $N=U^{-1}R$ the localization at $U$. Then the following two chain complexes are equivalent:
			\[
			\mathcal{C}(\Delta)\otimes_{\mathbb{Z}} N:
			\begin{tikzcd}[row sep=1.5cm,column sep=1.5cm]
				\cdots \arrow[r, "\delta_{k+1}\otimes id_N"]&C_k\otimes_{\mathbb{Z}} N \arrow[r, "\delta_{k}\otimes id_N"]& C_{k-1}\otimes_{\mathbb{Z}} N \arrow[r, "\delta_{k-1}\otimes id_{N}"] &\cdots 
			\end{tikzcd}
			\]
			\[ \mathcal{C}(\Delta,S)\otimes_{R} N:
			\begin{tikzcd}[row sep=1.5cm,column sep=1.5cm]
				\cdots \arrow[r, "\tilde{\delta}_{k+1}\otimes id_N"]&C_k(S)\otimes_{R} N \arrow[r, "\tilde{\delta}_{k}\otimes id_{N}"]& C_{k-1}(S)\otimes_{R} N \arrow[r, "\tilde{\delta}_{k-1}\otimes id_N"] &\cdots 
			\end{tikzcd}.
			\]
			And hence the homology $N-$modules are isomorphic.
			In particular, for the persistent version, we get that $\mathcal{C}(\Delta_\epsilon,S)\otimes_{R}N$ is equivalent to $\mathcal{C}(\Delta_{\epsilon})\otimes_{\mathbb{Z}}N$ for every $\epsilon$.
		\end{theorem}
		\begin{proof}
			Just similar to the proof of corollary~\ref{equivalence of the two chain complexes via evaluation}. Only note that by definition of the multiplicative set, $m_A\in U$ for all $A\in\Delta$ and thus invertible. $\frac{1}{m_A}$ here is understood as the element in the localization $N$.
		\end{proof}
		\begin{remark}\label{remark on the equivalent chain complexes}
			\begin{enumerate}[(a)]
				\item Both Theorem~\ref{equivalence of the two chain complexes via evaluation} and Theorem~\ref{equivalence of the two chain complexes via localization} hold for the reduced chain complexes and their persistent version according to the proof of the theorems.
				\item Theorem~\ref{equivalence of the two chain complexes via localization} contains the special case when $R$ is an integral domain $\mathbb{D}$ and $N$ is the quotient field $\Frac(\mathbb{D})$ of $\mathbb{D}$.
				\item Theorem~\ref{equivalence of the two chain complexes via evaluation} and Theorem~\ref{equivalence of the two chain complexes via localization} ensure us to read the information of the original chain complex over some specific ring or field from the associated complex. Hence it may be more promising to use the persistent chain complex of free modules in TDA.
			\end{enumerate}
		\end{remark}

	 For the case when $m_i\in\mathfrak{m}$ for the associated complex in Theorem~\ref{equivalence of the two chain complexes via evaluation} or the case when $m_i\not\in U$ for the  one in Theorem~\ref{equivalence of the two chain complexes via localization}, the map $\mathcal{A}_k$ is well-defined. But if we restrict our consideration on some subcomplex which makes the conditions valid, we can also get the equivalence between the subcomplex. That's a way to extract some "local" information of the complex as the following corollaries, which are directly got from Theorem~\ref{equivalence of the two chain complexes via evaluation} and Theorem~\ref{equivalence of the two chain complexes via localization}, imply:
	 \begin{corollary}[Local information, I]\label{corollary: local information via evaluation}
	 	Let $\mathfrak{m}$ be a maximal ideal of $R$ and $\mathbb{F}=R/\mathfrak{m}$ be the quotient ring (which is actually a field). Let 
	 	$W:=\{i\in\{1,\dots,n\}|m_i\not\in \mathfrak{m}\}$ and $T:=\{m_i|i\in W\}$.
	 	Denote the full subcomplex of $\Delta$ over the vertex set $W$ as $\Delta'$. 	Then  the (sub-)chain complex $\mathcal{C}(\Delta')\otimes_{\mathbb{Z}} \mathbb{F}$ and $\mathcal{C}(\Delta',T)\otimes_{R} \mathbb{F}$ are equivalent . 
	 	Moreover, in the persistent version, for every $\epsilon$, denote the full subcomplce $\Delta_\epsilon$ over $W$ as  $\Delta_\epsilon'$. Then for any $\epsilon$,  the (sub-)chain complexes $\mathcal{C}(\Delta_\epsilon')\otimes_{\mathbb{Z}} \mathbb{F}$ and $\mathcal{C}(\Delta_\epsilon',T)\otimes_{R} \mathbb{F}$ are always equivalent to each other. 
	 	
	 \end{corollary}
		
		\begin{corollary}[Local information, II]\label{corollary: local information via localization}
		Let $U$ be a multiplicative set of $R$ and $N=U^{-1}R$ the localization at $U$. Let 
			$W:=\{i\in\{1,\dots,n\}|m_i\in U\}$ and $T:=\{m_i|i\in W\}$.
			Denote the full subcomplex of $\Delta$ over the vertex set $W$ as $\Delta'$. 	Then  the (sub-)chain complexes $\mathcal{C}(\Delta')\otimes_{\mathbb{Z}} \mathbb{F}$ and $\mathcal{C}(\Delta',T)\otimes_{R} \mathbb{F}$ are equivalent. 
			Moreover, in the persistent version, for every $\epsilon$, denote the full subcomplex $\Delta_\epsilon$ over $W$ as  $\Delta_\epsilon'$. Then for any $\epsilon$,  the (sub-)chain complexes $\mathcal{C}(\Delta_\epsilon')\otimes_{\mathbb{Z}} N$  and $\mathcal{C}(\Delta_\epsilon',T)\otimes_{R} N$ are always equivalent. 
		\end{corollary}
		\begin{example}\label{eg: example of label complex}
			Consider the following labelled complex $\Delta$ with label elements  $S_0=\{m_i|1\leq i\leq 3\}$ in $R=\mathbb{C}[x_1,x_2]$ as follows:
			\[
			\begin{tikzpicture}
				\draw[blue,thick] (0,1)--(-1,0);
				\draw[blue,thick] (0,1)--(1,0);
				\draw[blue,thick] (-1,0)--(1,0);
				\filldraw [black] (0,1) circle (2pt)node[anchor=west]{$m_1=x_2+x_1$};
				\filldraw [black] (1,0) circle (2pt)node[anchor=west]{$m_3=x_1x_2$};
				\filldraw [black] (-1,0) circle (2pt)node[anchor=east]{$m_2=x_1$};
			\end{tikzpicture}
			\]
			Let $
			\mathcal{B}_1:=\{\{1,2\},\{2,3\},\{1,3\}\}$ be the canonial basis of the chain groups $C_1(S_0)$ and $C_1$,  and $
			\mathcal{B}_0:=\{\{1\},\{2\},\{3\}\}$ the canonical basis of the chain groups $C_0(S_0)$ and $C_0$.
			Then with respect to these bases, the matrices of the boundary maps $\tilde{\delta}_1$ and $\delta_1$ are
	$ \tilde{Y}_1:=\begin{pmatrix}
					x_1+x_2&0&x_1+x_2\\
					-x_1&1&0\\
					0&-x_2&x_1x_2
					\end{pmatrix}$ 
			and
				$Y_1:=\begin{pmatrix}
					1&0&1\\
					-1&1&0\\
					0&-1&-1\\
				\end{pmatrix}$ respectively.
				
			By Hilbert's Nullstellensatz\cite[Chapter 4]{cox1997ideals}, $\mathfrak{m}:=\langle x_1-1,x_2-1\rangle_{R}$ is a maximal ideal of $R$ and $m_i\not\in\mathfrak{m}, 1\leq i\leq 4$. Thus $\mathbb{F}:=R/\mathfrak{m}$ and $\mathcal{S}:=R_{\mathfrak{m}}=U^{-1}R$, where $U=R\setminus \mathfrak{m}$, satisfies the conditions of Theorem~\ref{equivalence of the two chain complexes via evaluation} and Theorem~\ref{equivalence of the two chain complexes via localization} respectively. 
			Let 
			$D_0:=\diag([x_1+x_2],[x_1],[x_1x_2])$ and 
			$D_1:=\diag([x_1(x_1+x_2)],[x_1x_2],[x_1x_2(x_1+x_2)])$
			where the square brackets means the equivalent class in $\mathbb{F}$. Via easy calculation, we get that $D_0^{-1}[Y_1]D_1=[\tilde{Y}_1]$.
			 Here $[Y_1]$ and $[\tilde{Y}_1]$ are the matrices of the $\mathbb{F}-$linear mappings $\tilde{\delta}_1\otimes id_{\mathbb{F}}$ and $\delta_1\otimes id_{\mathbb{F}}$  under the canonical $\mathbb{F}-$bases $\{\sigma\otimes 1|\sigma\in\mathcal{B}_1\}$  and $\{\sigma\otimes 1|\sigma\in\mathcal{B}_{0}\}$ respectively. 
			 This verifies Theorem~\ref{equivalence of the two chain complexes via evaluation}. Note that in this case, the matrix of the mapping $\mathcal{A}_1$, defined in the Theorem~\ref{equivalence of the two chain complexes via evaluation}, under the canonical basis is  $D_1^{-1}$. In a similar way, we can validate Theorem~~\ref{equivalence of the two chain complexes via localization}.
			 
			 Now consider $\mathfrak{m}_1:=\langle x_1-1,x_2+1\rangle_{R}$, then by Hilbert's Nullstellensatz\cite[Chapter 4]{cox1997ideals}, $\mathfrak{m}_1$ is a maximal ideal of $R$, and $m_2,m_3\not\in\mathfrak{m}_1$ but $m_1\in\mathfrak{m}_1$.
			 Take consideration of $\mathcal{C}(\Delta_1)\otimes_{\mathbb{Z}}R_{\mathfrak{m}_1}$ and $\mathcal{C}(\Delta_1,S_0)\otimes_{R}R_{\mathfrak{m}_1}$, then 
			 the matrix of $\delta_1\otimes id_{R_{\mathfrak{m}_1}}$ under the canonical basis $\{\sigma\otimes 1|\sigma\in\mathcal{B}_1\}$ and $\{\sigma\otimes 1|\sigma\in\mathcal{B}_0\}$ is $W_{1}:=[Y_1]_1$ while the matrix of $\tilde{\delta}_1\otimes id_{R_{\mathfrak{m}_1}}$ is $\tilde{W}_{1}:=[\tilde{Y}_1]_1$,where $[]_1$ denotes the equivalence class in $R_{\mathfrak{m}_1}$.
			 
			 Note that in this case $D_1':=\diag([x_1(x_1+x_2)]_1,[x_1x_2]_1,[x_1x_2(x_1+x_2)]_1)$ is not invertible, since $\det(D_1')=[x_1]_1^3[x_2(x_2+x_1)]_1^2$ is not a unit in the integral domain $R_{\mathfrak{m_1}}$. So we can't obtain the mapping $\mathcal{A}_1$ defined in the Theorem~\ref{equivalence of the two chain complexes via localization}.
			 
			 But if we consider the full subcomplex $\Delta_2$ of $\Delta_1$ defined by the vertex set $\{2,3\}$, then the condition of Corollary~\ref{corollary: local information via localization} is satisfied, and we can verify that the result of Corollary~\ref{corollary: local information via localization} holds in the same way as above. This enables us to extract the "local" information of $\Delta_2$.
		\end{example}

				\subsection{More properties of persistent chain complex of free modules when $R=\mathbb{F}[x_1,\dots,x_t]$}	
					
					Now let's focus on the case when $R=\mathbb{S}_t=\mathbb{F}[x_1,\dots,x_t]$. Since $\mathbb{S}_t$ is an $\mathbb{N}^t$-graded ring and thus a $\mathbb{Z}^t-$graded ring as in the the section~\ref{appendix:graded rings} of the Appendix, we expect to get more information from the associated complex and that's indeed possible. Because we need to use the graded structure, then we always assume that the label elements $S:=\{m_1,\dots,m_n\}$ are all \emph{monomials} in $\mathbb{S}_t$. And in the following, $\deg(m)$ for a monomial $m$ denotes the $\mathbb{N}^t$ graded degree of $m$, namely, for $\alpha\in\mathbb{N}^t$ and $m=x^{\alpha}$, $\deg(m)=\alpha$.
					
					Endow the $\mathbb{S}_t-$modules $C_k(S)$ with an $\mathbb{N}^t-$graded structure as follows: we assume $\deg(\sigma):=\deg(m_\sigma)\in\mathbb{N}^t$ for  $\sigma$ a $k-$dimensional face in $\Delta$.   Then $C_{k}(S)=\bigoplus_{\alpha\in \mathbb{N}^t}(C_{k}(S))_{\alpha}$, where
					\begin{equation}\label{degree alpha part of the chain groups}
						(C_{k}(S))_{\alpha}:=\bigoplus_{\sigma\in\Delta,\#\sigma=k+1}\spn_{\mathbb{F}}\{q_\sigma\sigma|\deg(\sigma)+\deg(q_{\sigma})=\alpha, q_{\sigma} \text{ is a monomial.}\}
					\end{equation}
					makes $C_k(S)$ an $\mathbb{N}^t-$graded $\mathbb{S}_t-$modules.
					 Then it is easy to see from the equation~\eqref{the definition of the boundary map of the associated complex} that the boundary map $\tilde{\delta}$ is of degree zero. For $m\in\mathbb{S}_t$ and $m$ monomial, we define $\Delta_m$ to be the subcomplex of $\Delta$ which contains all the faces whose label elemnet divides $m$. Note that by definition, $\Delta_m$ is always a full subcomplex over certain subset of vertices. By considering each homogeneous part, we can get the following theorem:
					\begin{theorem}[local indormation, III ]\label{local information}
						Let $\Delta$ be a simplicial complex over $\{1,\dots,n\}$ and $S_0:=\{m_1,\dots,m_n\}$ the set of label elements.  For any $\alpha\in\mathbb{N}^t$, let $m_\alpha=x^\alpha$.
						Then the following  reduced complex of $\mathbb{K}-$vector space:
						\[
						\mathcal{C}(\Delta,S_0)_{\alpha}:
						\begin{tikzcd}
							\cdots \arrow[r,"\tilde{\delta}_{k+1}"]&(C_{k}(S_0))_\alpha \arrow[r,"\tilde{\delta}_k"]& (C_{k-1}(S_0))_\alpha \arrow[r,"\tilde{\delta}_{k-1}"]&\cdots
						\end{tikzcd}
						\]
						where $(C_k(S_0))_\alpha=\{X\in C_k(S_0)|\deg(X)=\alpha\}$ is the degree-$\alpha$ part of $C_{k}(S_0)$,  and the reduced subcomplex $\mathcal{C}(\Delta_{m_\alpha})$ with coefficients in $\mathbb{K}$ are equivalent as chain complex of $\mathbb{F}-$vector spaces. 
					In particular, the homology groups of the two complexes are isomorphic.
					Furthermore, for the persistent version, the choice of $m_\alpha$ is independent of the parameter $\epsilon$ in the persistent associated chain complex of free modules, thus we can always have the equivalence between $\mathcal{C}(\Delta_\epsilon,S_0)_{\alpha}$ and $(\Delta_\epsilon)_{m_\alpha}$ for any $\epsilon$.
					\end{theorem}
					\begin{proof}
						We only need to prove the first assertion and the second one follows. The proof of the first one is essentially inspired by the proof of \cite[Theorem 2.2]{eisenbud2005syzygies}. 
						
						 By the definition of the graded structure on $C_{k}(S_0)$ as Eq~\eqref{degree alpha part of the chain groups} shows, we can get that 
						 \begin{equation}
						 	(C_{k}(S_0))_{\alpha}=\bigoplus_{\sigma\in\Delta, \#\sigma =k+1, m_{\sigma}|m_{\alpha}} \spn_\mathbb{F}\{\frac{m_\alpha}{m_{\sigma}}\sigma\},
						 \end{equation}
						By definition of $\Delta_{m_{\alpha}}$, we see that the $\mathbb{F}-$vector space $(C_{k}{S_0})_{\alpha}$ is isomorphic to the $\mathbb{F}-$vector space $C_{k}(\Delta_{m_\alpha})$, and the isomorphism can be taken as 
						\begin{equation}
							\begin{aligned}
								\mathcal{F}_{k}:C_{k}(\Delta_{m_\alpha})&\longrightarrow (C_{k}{S_0})_{\alpha}\\
								\sum_{\sigma}l_\sigma \sigma\mapsto & \sum_{\sigma}l_\sigma \frac{m_\alpha}{m_{\sigma}} \sigma
							\end{aligned}
						\end{equation}
					It is easy to check that $\mathcal{F}_{k-1}\circ \delta_{k}=\tilde{\delta}_{k}\circ\mathcal{F}_{k}$, so the family of maps $\{\mathcal{F}_k\}$ can be taken as the chain map between $\mathcal{C}(\Delta,S_0)_{\alpha}$ and $\mathcal{C}(\Delta_{m_{\alpha}})$ with coefficients in $\mathbb{F}$. 
					\end{proof}
					\begin{remark}
						\begin{enumerate}[(a)]
							\item Theorem~\ref{local information} enables us to extract information of the \emph{full subcomplex} of $\Delta$ over some given subset of the vertex set, including the whole complex which can be got by choosing $\alpha$ as the degree of the least common multiplier of all the $m_i$, by considering the homogeneous part of the persistent chain complex of free modules. In fact, we can choose the label polynomials so that for any subset $W$ of the vertex set, we can choose a degree $\alpha\in\mathbb{N}^t$ such that $\Delta_{m_{\alpha}}$ is exactly the full subcomplex on $W$. For instance, we choose the label polynomial $m_i$ to be $x_i$, then if we want to extract the full subcomplex of $\Delta$ over the subset $W$ of the vertex set, we just choose the indicator vector $\alpha_{W}$ defined by $\alpha_i=\begin{cases}
								1& \text{ if }i\in W\\
								0&  \text{ if }i\not\in W
							\end{cases}$, and the subcomplex $\Delta_{m_{\alpha_W}}$ is just what we want.
							\item Actually by the proof of the theorem, we can replace all the reduced complex as the ordinary complex. See the following example.
						\end{enumerate}
					\end{remark}
					\begin{example}
						Consider the simplicial complex 
						$$\Delta:=\{\{1\},\{2\},\{3\},\{4\},\{1,2\},\{1,3\},\{2,3\},\{1,2,3\},\{1,4\}\}$$ with the vertex set $\{1,2,3,4\}$ with the label polynomials $m_1=x_1,x_2=x_2x_3,x_3=x_2x_4,m_4=x_3x_4$, as the picture shows:  \[\begin{tikzpicture}
							\draw[blue,thick] (0,0)--(0,1);
							\filldraw[gray] (0,0)--(-1,-1)--(1,-1)--cycle;
							\filldraw [black] (0,1) circle (2pt)node[anchor=west]{$m_4=x_3x_4$};
							\filldraw [black] (0,0) circle (2pt)node[anchor=west]{$m_1=x_1$};
							\filldraw [black] (1,-1) circle (2pt)node[anchor=west]{$m_3=x_2x_4$};
							\filldraw [black] (-1,-1) circle (2pt)node[anchor=east]{$m_2=x_2x_3$};
						\end{tikzpicture}.\]
						Under the canonical basis, the matrices of the boundary maps can be calculated as follows:
						$$\tilde{\delta}_2\leadsto \begin{pmatrix}
							-x_4\\x_3\\-x_1\\0
						\end{pmatrix}\qquad 
						\tilde{\delta}_1\leadsto \begin{pmatrix}
							x_2,x_3&x_2x_4&0&x_3x_4\\
							-x_1&0&x_4&0\\
							0&-x_1&-x_3&0\\
							0&0&0&-x_1\\
						\end{pmatrix}$$
						Choose $\alpha=(0,1,1,1)$, then $m_\alpha=x_2x_3x_4$; so
						{\small \[ \mathcal{C}(\Delta_{m_\alpha}):
							\begin{tikzcd}
								0\arrow[r,"0"]&\spn_{\mathbb{F}}(\{23\})\arrow[r,"\delta_1"]&\spn_{\mathbb{F}}(\{2,3,4\})\arrow[r,"\delta_0"]&\spn_{\mathbb{F}}(\emptyset)\arrow[r,"0"]&0
							\end{tikzcd}
							\]}
						and 
						$\delta_1\leadsto \begin{pmatrix}
							1\\-1\\0
						\end{pmatrix}$\quad $\delta_0\leadsto\begin{pmatrix}
							-1&-1&-1
						\end{pmatrix}$. 
						On the other hand, degree-$\alpha$ part of the associated complex is as follows
						{\small \[
							\begin{tikzcd}
								0\arrow[r,"0"]&\spn_{\mathbb{F}}(\{2,3\})\arrow[r,"\tilde{\delta}_1"]&\spn_\mathbb{F}(\{x_2\{4\},x_3\{3\},x_4\{2\}\})\arrow[r,"\tilde{\delta}_0"]&\spn_{\mathbb{F}}(\{x_2x_3x_4\emptyset\})\arrow[r,"0"]&0
							\end{tikzcd}
							\]}
						and $\tilde{\delta}_1\leadsto \begin{pmatrix}
							1\\-1\\0
						\end{pmatrix}$\quad $\tilde{\delta}_0\leadsto\begin{pmatrix}
							-1&-1&-1
						\end{pmatrix}$
						which are the same as the ones above.
					\end{example}

					\bibliographystyle{plain}
					\bibliography{ref}
				\appendix	\section{Commutative algebra dictionary}
					In this Appendix, we will briefly introduce basic knowledge from commutative algebra, mainly about the properties of the basic objects: the ideals in a commutative ring $R$ and the $R$-modules. 
					Restricted to the length of the paper, many details will be omitted.   There are abundant references to commutative algebra for those who are interested, such as \cite{Atiyah,eisenbud2013commutative,Kemper2011,miller2006combinatorial,eisenbud2005syzygies}. And \cite{cox1997ideals,cox2003using} are good references for the computational methods of commutative algebra. 
					
					Let $R$ be a ring and $M$ an $R-$module. The basic examples of modules we should bear in mind are the $\mathbb{F}-$vector spaces over a field $\mathbb{F}$ (as $\mathbb{F}-$modules), the abelian groups as $\mathbb{Z}-$modules and the ring $R$ itself as an $R-$module. For the last example, we emphasize that the submodules are merely the ideals of $R$ and the quotient modules are simply the quotient rings by some ideal. Also notice that for any ring $R$, we can "forget" the structure of multiplication and view it as a $\mathbb{Z}-$module for the abelian group defined by the structure of addition.
					We will use $ \langle S\rangle_{M}$ or $\spn_R(S)$ for the submodule generated by $S\subset M$.

						\subsubsection{Localization}
						An important and useful tool in the ring/module theory is the \emph{localization}. We'll only introduce the localization of a ring $R$. For a ring $R$, a subset $U\subset R$ is called \emph{a multiplicative set} if $1\in U$ and $U$ is \emph{closed} under the multiplication. Now suppose $U$ is a multiplicative set of $R$. The \emph{localization 
							$U^{-1}R$ of $R$ at $U$} is the quotient set $U^{-1}R:=\{(a,b)|a\in R,b\in U\}/\sim$, where the equivalence relation $\sim$ is defined by $(a_1,b_1)\sim(a_2,b_2)$ if and only if $\exists u\in U$, such that $u(a_1b_2-a_2b_1)=0$. Define the addition and multiplication on $U^{-1}R$ as:
						\begin{equation}\label{addition and multiplication in the localization}
							\begin{aligned}
								(a_1,b_1)+(a_2,b_2):=(a_1b_2+a_2b_1,b_1b_2);\\
								(a_1,b_1)\cdot (a_2,b_2):=(a_1a_2,b_1b_2).
							\end{aligned}
						\end{equation}
						One can check that under this addition and multiplication, $U^{-1}R$ is a commutative ring with unit--the equivalence class of $(1,1)$. We also often write the equivalence class of $(a,b)$ as $\frac{a}{b}$. A familiar example is that for an integral domain $\mathbb{D}$ like $\mathbb{Z}$, $U=\mathbb{D}\setminus \{0\}$, then $U^{-1}\mathbb{D}$ is just the \emph{quotient field or field of fractions} of $\mathbb{D}$. We usually denote it by $\Frac(\mathbb{D})$.  Another typical example is that for a proper \emph{prime ideal} $P$ of $R$, the complementary set $P^c$ of $P$ is a multiplicative set, and thus we can do the localization at $P^c$. In this case, we will denote the localization as $R_{P}$. 
						\subsubsection{Primary and prime decompositions}

						Prime and primary ideals play a central role in commutative algebra. Specifically, a proper ideal $P$ of $R$ is \emph{prime} if $R/P$ is an \emph{integral domain}, i.e., the only zero-divisor of $R/P$ is the zero element.  More generally, a proper ideal $Q$ of $R$ is \emph{primary} if each zero-divisor in $R/Q$ is \emph{nilpotent}. In particular, every prime ideal is primary. The \emph{radical ideal} $\sqrt{I}$ of an ideal $I \subset R$ is defined as the ideal $\sqrt{I}:=\{ u \in R \mid u^n\in I \text{ for some } n\in\mathbb{Z}_+ \}$. It is obvious that $I \subset \sqrt{I}$. Additionally, an ideal $I$ is said to be \emph{radical} if $\sqrt{I}=I$. By definition, every prime ideal $P$ of a ring is radical, and the radical ideal $\sqrt{Q}$ of a primary ideal $Q$ is prime.
						
						An ideal $I$ of $R$ is called \emph{primary decomposable} if there exist finitely many distinct primary ideals $Q_1, \dots, Q_n$ such that $I = \bigcap_{i = 1}^n Q_i$. Furthermore, this primary decomposition is said to be \emph{minimal} (or \emph{reduced}) if it satisfies the following properties: (a) the radical ideals $\sqrt{Q_i}$ are all distinct and (b) $\bigcap_{j \neq i} Q_j \nsubseteq Q_i$ for every $i \in \{ 1, 2, \dots, n \}$. Similarly, an ideal $J$ of $R$ is called \emph{prime decomposable} if there exist finitely many distinct prime ideals $P_1, \dots, P_n$ such that $J = \bigcap_{i = 1}^n P_i$, and this prime decomposition is said to be \emph{minimal} (or \emph{reduced}) if  $\bigcap_{j \neq i} P_j \nsubseteq P_i$ for each $i \in \{ 1, 2, \dots, n \}$. Evidently, a prime decomposition is also a primary decomposition since every prime ideal is primary. Moreover, any primary decomposition \( \bigcap_{i = 1}^{n} Q_i \) always admits a minimal one by removing redundant primary ideals from the decomposition (cf. \cite[Chapter 3]{eisenbud2013commutative} or \cite[Chapter 4]{Atiyah}). Furthermore, in a more general setting, both decompositions can be extended to \( R \)-modules (cf. \cite[Chapter 3]{eisenbud2013commutative}). 
						
						It is well known that an ideal in a ring can have different minimal primary decompositions corresponding to different families of primary ideals \cite[Chapter 4]{Atiyah}. However, the following proposition shows that by examining the family of so-called \emph{associated primes} induced by a minimal primary decomposition, the primary decomposition exhibits a form of uniqueness based on the associated prime ideals. In particular, every prime decomposable ideal in a ring admits a unique minimal prime decomposition.
						\begin{proposition}\label{Strong uniqueness of prime decomposition}
							Let $R$ be a ring, $I$ be a primary decomposable ideal, and 
							$I = \bigcap_{i = 1}^n Q_i$ be a minimal primary decomposition of $I$. 
							Then the set of prime ideals $\Ass(I):=\{\sqrt{Q_i}\}_{1\leq i\leq n}$, called the set of \emph{associated primes} of $I$, is independent of the minimal primary decomposition of $I$. 
							
							Furthermore, if $I$ admits a minimal prime decomposition $I=\bigcap_{i=1}^n P_i$ with prime ideals $P_i$, then $\{ P_i \}_{1\leq i\leq n}$ is exactly the set of associated primes of $I$. In other words, if $I$ is prime decomposable, then the minimal prime decomposition of $I$ is unique. 
						\end{proposition}
						\begin{proof}
							The first assertion corresponds to the first uniqueness proposition in \cite[Chapter 4]{Atiyah}. The second follows directly from the definition, the fact that every prime ideal is primary, and the first assertion.
						\end{proof}
						
						Every primary ideal $Q$ of a ring $R$ is primary decomposable by definition, with $\Ass(Q) = \{\sqrt{Q}\}$. Similarly, every prime ideal $P$ is prime decomposable, and $\Ass(P) = \{P\}$. Additionally, for a radical ideal, the associated primes are minimal in the following sense: 
						
						\begin{proposition}[Proposition 4.6 and Exercise 2, Chapter 4, \cite{Atiyah}]\label{Proposition: Associated primes are precisely the minimal elements}
							Let $I$ be a primary decomposable ideal in a ring $R$ and $I$ be radical. Then the associated primes of $I$ are precisely the minimal elements in the set of all prime ideals containing $I$.
						\end{proposition}
						
						In general, an ideal may not always be prime or primary decomposable. However, certain properties of the ring and the ideal can guarantee the existence of such decompositions. For example, as shown below, in a \emph{Noetherian} ring, every ideal admits a primary decomposition (Proposition \ref{Existence of primary decomposition}).

						

						\subsubsection{Noetherian rings }

						A ring $R$ is called \emph{Noetherian} if every ideal of it is finitely generated. The Noetherian property frequently appears in rings studied in commutative algebra, such as \emph{principal ideal domains} (PIDs), including classical examples like \(\mathbb{Z}\) and the univariate polynomial ring \(\mathbb{F}[x]\) over a field \(\mathbb{F}\). More generally, every multivariate polynomial ring \( R[x_1, x_2, \dots, x_n]\) over a Noetherian ring \(R\) is itself Noetherian, as established by the well-known Hilbert basis theorem (see, for example, \cite[Chapter 7]{Jacobsonn1989vol2}). Notably, polynomial rings $\mathbb{Z}[x_1,\dots,x_n]$ and $\mathbb{F}[x_1,\dots,x_n]$, where $\mathbb{F}$ is a field, are Noetherian. The following proposition indicates that every proper ideal in a Noetherian ring has primary decompositions.
						
						\begin{proposition}[Theorem 3.10, Chapter 3, \cite{eisenbud2013commutative}]\label{Existence of primary decomposition}
							Let $R$ be a Noetherian ring and $I$ a proper ideal of $R$. Then, $I$ admits a primary decomposition.
						\end{proposition}
						
						\begin{proposition}\label{Equivalent condition for prime decomposable ideals in a Noetherian ring}
							Let $R$ be a Noetherian ring and $I$ be a proper ideal of $R$, then $I$ has a prime decomposition if and only if $I$ is a radical ideal.
						\end{proposition}
						\begin{proof}
							Recall that the radical ideal of $\bigcap_{j = 1}^m I_j$ of ideals $I_1, I_2, ..., I_m$ in $R$ equals the intersection $\bigcap_{j = 1}^m \sqrt{I_j}$ (Exercise 1.13, \cite{Atiyah}). Suppose $I$ is a radical ideal. Then, by Theorem \ref{Existence of primary decomposition}, $I$ can be expressed as $I = \bigcap_{i=1}^n Q_i$, where each $Q_i$ is a primary ideal. Consequently, $I = \sqrt{I} = \bigcap_{i=1}^n \sqrt{Q_i}$. Since each $Q_i$ is a primary ideal for \( i \in \{ 1, 2, \dots, n \} \), it follows that each $\sqrt{Q_i}$ is a prime ideal. This demonstrates that $I$ admits a prime decomposition. Conversely, if $I = \bigcap_{i=1}^n P_i$, where each $P_i$ is a prime ideal, then $\sqrt{I} = \bigcap_{i=1}^n \sqrt{P_i} = \bigcap_{i=1}^n P_i = I$. In other words, $I$ is a radical ideal in $R$.
						\end{proof}

						\subsubsection{Graded rings and modules}\label{appendix:graded rings}
						Let $(\Lambda, +, 0)$ be a commutative, additive monoid (e.g., $\Lambda = \mathbb{N}^n$). A ring $R$ is called an \emph{$\Lambda$-graded ring} if it admits a \emph{graded structure} indexed by $\Lambda$, i.e., $R$ is a direct sum $R=\bigoplus_{\alpha\in\Lambda} R_\alpha$ as \emph{abelian groups} with the property that $R_\alpha R_\beta \subset R_{\alpha + \beta}$ whenever $\alpha, \beta \in \Lambda$. Similarly, a module $M$ over a graded ring $R=\bigoplus_{\alpha\in\Lambda} R_\alpha$ is referred to a \emph{graded $R$-module} if $M = \bigoplus_{\alpha \in \Lambda} M_\alpha$ is a direct sum as \emph{abelian groups} so that $R_\alpha M_\beta \subset M_{\alpha + \beta}$ whenever $\alpha, \beta \in \Lambda$. 
						
						A nonzero element in $R_\alpha$ (resp. $M_\alpha$) is called a \emph{homogeneous elements of degree $\alpha$} , and the abelian group $R_\alpha$ (resp. $M_\alpha$) is called the  \emph{homogeneous part of $R$} (resp. $M$) \emph{of degree $\alpha$}. A \emph{homogeneous ideal} (resp. \emph{homogeneous submodule}) of a graded ring (resp. graded module) is an ideal (resp. submodule) that is generated by homogeneous elements. Notice that 
						\begin{equation*}
							R/I \simeq \bigoplus_{\alpha \in \Lambda} R_\alpha/I_\alpha \text{ and } M/N \simeq \bigoplus_{\alpha \in \Lambda} M_\alpha/N_\alpha   
						\end{equation*}
						for any homogeneous ideal $I = \bigoplus_{\alpha \in \Lambda} I_\alpha$ of $R = \bigoplus_{\alpha \in \Lambda} R_\alpha$ and homogeneous submodule $N = \bigoplus_{\alpha \in \Lambda} N_\alpha$ of a graded $R$-module $M = \bigoplus_{\alpha \in \Lambda} M_\alpha$; in other words, the quotient ring (resp. module) of a homogeneous ideal (resp. submodule) inherits a naturally graded structure from the original one.
						
							
						%
						\begin{example}
							
							A typical example is the multivariant polynomial ring $\mathcal{S}:=R[x_1,\dots,x_n]$ over a ring $R$. It can be viewed as $\mathbb{N}-$graded ring via the decomposition $\mathcal{S}=\bigotimes_{i\in \mathbb{N}}\mathcal{S}_i$ with $\mathcal{S}_i:=\{f\in\mathcal{S}|f\text{ is a homogenous polynomial of degree } i\}\cup \{0\}$. Another natural and important graded structure on $\mathcal{S}$ is the $\mathbb{N}^n-$graded structure: 
							$\mathcal{S}=\oplus_{\alpha\in \mathbb{N}^n}\mathcal{S}_\alpha$, where $\mathcal{S}_{\alpha}:=\spn_{R}(x^\alpha)$. 
						\end{example}
						For an $\mathbb{N}^n-$graded module $M$, usually by convention, we can assume that $M_\alpha=0$ for all $\alpha\in \mathbb{Z}^n\setminus \mathbb{N}^n$ and then view $M$ as a $\mathbb{Z}^n-$graded module. 
						A homomorphism between two $\mathbb{Z}^n-$graded modules $\delta:M\longrightarrow N$ is called \emph{a graded homomorphism of (homogeneous) degree $\alpha\in \mathbb{Z}^n$} if $\delta(M_\beta)\subset N_{\alpha+\beta}$ for all $\beta\in \mathbb{Z}^n$.  
						Sometimes for simplicity, we may use the \emph{twisted graded module} $N(\alpha)$ defined by $N(\alpha)_\beta:=N_{\alpha+\beta}$ to replace $N$ and force the degree of the map to be $0$.
						The same consideration can be applied to an $\mathbb{N}^n-$graded or $\mathbb{Z}^n-$graded ring.
						A common example is that the usual differential operator $\frac{d}{dt}:\mathbb{F}[t]\longrightarrow \mathbb{F}[t]$ over the polynomial ring $\mathbb{F}[t]$ is graded homomorphism of degree $-1$.
						
						The primary/prime decomposition of ideals/submodules can be extended to the graded case, see for example, \cite[Chapter 3]{eisenbud2013commutative} and \cite[Chapter 6]{cox2003using}.

						\subsubsection{Monomial ideals.} An ideal $I$ of $\mathbb{S}_n:=\mathbb{F}[x_1,\dots,x_n]$, with $\mathbb{F}$ a field, is called a \emph{monomial ideal} if it has a set of generators consisting of \emph{only} monomials.  The class of monomial ideals is very important for both theoretical research and computational issues due to the following simple but good properties:
						\begin{lemma}[Criterion for members of monomial ideals]\cite[Lemma 2, Chapter 2]{cox1997ideals}\label{Criterion for members of monomial ideals}
							Let $\Lambda\subset \mathbb{N}^n$ and $I = \langle x^{\mathbf{u}} \mid \mathbf{u} \in \Lambda \rangle_{\mathbb{S}_n}$. Then $x^{\mathbf{v}} \in I$ if and only if $\mathbf{u} \leq \mathbf{v}$ for some $\mathbf{u} \in \Lambda$. 
						\end{lemma}
						\begin{lemma}[Minimal basis of monomial ideals]\cite[Proposition 7, Chapter 2]{cox1997ideals}\label{minimal basis of monomial ideals}
							Every monomial ideal \( I \) of  \( \mathbb{S}_n \) admits a family of generators \( \Lambda = \{ x^{\mathbf{u}_1}, \dots, x^{\mathbf{u}_m} \} \) of monomials, called a \emph{minimal basis} of \( I \), which satisfies that \( x^{\mathbf{u}_i} \nmid x^{\mathbf{u}_j} \) for all \( i \neq j \). 
						\end{lemma}
						\begin{lemma}[Radical ideal of a monomial ideal]\cite[Proposition 1.2.4 \& Corollary 1.2.5]{Hibi2011monomial}\label{radical monomial ideal}
							
							Let $I=\langle m_1,\dots,m_t\rangle_{\mathbb{S}_n}$ be a monomial ideal with $m_1,\dots,m_t$ monomials, then $\sqrt{I}=\langle\Sqrt(m_i)\rangle_{\mathbb{S}_n}$, where $\Sqrt(m_i):=x_{i_1}\dots x_{i_k}$ if $m_i=x_{i_1}^{\alpha_1}\dots x_{i_k}^{\alpha_k}$ with $\alpha_{j}\in \mathbb{Z}_+,j\in \{i_1,\dots,i_k\}$.
							
							In particular, a monomial ideal $I$ is radical if and only if it's generated by \emph{square-free monomials}, i.e., monomials $m$ satisfying $m=\Sqrt(m)$. 
						\end{lemma}
						\begin{lemma}[Minimal primary decomposition of a monomial ideal]\cite[Theorem 1.3.1]{Hibi2011monomial}\label{primary decomposition of a radical monomial ideal}
							
							Let $I\subset \mathbb{S}_n$ be a monomial ideal. Then $I$ can be decomposed as $I=\bigcap_{i=1}^l Q_i$
							where each primary ideal $Q_i$ is of the form $\langle x_{i_1}^{d_1},x_{i_2}^{d_2},\dots, x_{i_s}^{d_s}\rangle$. Moreover, a minimal primary decomposition of this form is unique.
						\end{lemma}
						\begin{corollary}[Prime decomposition of a radical monomial ideal]\label{prime decomposition of a radical monomial ideal}
							
							Any radical monomial ideal \( I \subset \mathbb{S}_n\) admits a prime decomposition of the form:
							$I = \bigcap_{i=1}^{l} P_i$,
							where each \( P_i \) is a \emph{linear prime ideal} of the form $P_i = \langle  x_{i_1},x_{i_2},\dots, x_{i_s} \rangle$.
						\end{corollary} 

				\end{document}